\documentclass[a4paper,UKenglish,cleveref, autoref, thm-restate]{lipics-v2021}
\usepackage{graphicx,amsmath,amssymb,amsthm,mathtools,url}
\usepackage{hyperref,lineno,verbatim,soul,tabulary}
\newcolumntype{C}[1]{>{\centering\arraybackslash}p{#1}}

\theoremstyle{definition}
\newtheorem{dfn}[theorem]{Definition}
\newtheorem{pro}[theorem]{Problem}
\newtheorem{prop}[theorem]{Proposition}

\newtheorem{lem}[theorem]{Lemma}
\newtheorem{exa}[theorem]{Example}

\newcommand{\R}{\mathbb R}
\newcommand{\Z}{\mathbb Z}
\newcommand{\al}{\alpha}
\newcommand{\be}{\beta}
\newcommand{\ga}{\gamma}
\newcommand{\Ga}{\Gamma}
\newcommand{\de}{\delta}

\newcommand{\la}{\lambda}
\newcommand{\La}{\Lambda}

\newcommand{\ep}{\varepsilon}

\newcommand{\iso}{\mathrm{Iso}}
\newcommand{\SO}{\mathrm{SO}}
\newcommand{\RMS}{\mathrm{RMS}}

\newcommand{\AMD}{\mathrm{AMD}}

\newcommand{\vol}{\mathrm{Vol}}

\newcommand{\dif}{\mathrm{Dif}}

\newcommand{\bs}{\hfill $\blacksquare$}

\newcommand{\lra}{\leftrightarrow}
\newcommand{\vl}{\,:\,}

\title{The asymptotic behaviour and a near linear time algorithm for isometry invariants of periodic sets} %TODO Please add

\titlerunning{A near linear time algorithm for isometry invariants of periodic sets} %TODO optional, please use if title is longer than one line

\author{Daniel Widdowson}{Department of Computer Science, University of Liverpool, Liverpool, United Kingdom}{D.E.Widdowson@liverpool.ac.uk}{https://orcid.org/0000-0002-5958-0703}{}

\author{Marco Mosca}{Department of Computer Science, University of Liverpool, Liverpool, United Kingdom}{m.m.mosca@liverpool.ac.uk}{}{}

\author{Angeles Pulido}{Cambridge Crystallographic Data Centre, Cambridge, United Kingdom}{apulido@ccdc.cam.ac.uk}{0000-0002-7596-7262}{}

\author{Vitaliy Kurlin}{Department of Computer Science, University of Liverpool, Liverpool, United Kingdom}{vitaliy.kurlin@liverpool.ac.uk}{https://orcid.org/0000-0001-5328-5351}{EPSRC grant `Application-driven Topological Data Analysis', EP/R018472/1}

\author{Andrew I Cooper}{Materials Innovation Factory, University of Liverpool, Liverpool, United Kingdom}{aicooper@liverpool.ac.uk}{}{ERC Synergy Grant `Autonomous Discovery of Advanced Materials'}

\authorrunning{D. Widdowson and M. Mosca and A. Pulido and V. Kurlin and A. Cooper} %TODO mandatory. First: Use abbreviated first/middle names. Second (only in severe cases): Use first author plus 'et al.'

\Copyright{Daniel Widdowson and Marco Mosca and Angeles Pulido and Vitaliy Kurlin and Andrew Cooper} %TODO mandatory, please use full first names. LIPIcs license is "CC-BY";  http://creativecommons.org/licenses/by/3.0/

%TODO mandatory: Please choose ACM 2012 classifications from https://dl.acm.org/ccs/ccs_flat.cfm 

%TODO mandatory; please add comma-separated list of keywords

\linenumbers %uncomment to disable line numbering

\hideLIPIcs  %uncomment to remove references to LIPIcs series (logo, DOI, ...), e.g. when preparing a pre-final version to be uploaded to arXiv or another public repository

\EventEditors{John Q. Open and Joan R. Access}
\EventNoEds{2}
\EventLongTitle{42nd Conference on Very Important Topics (CVIT 2016)}
\EventShortTitle{CVIT 2016}
\EventAcronym{CVIT}
\EventYear{2016}
\EventDate{December 24--27, 2016}
\EventLocation{Little Whinging, United Kingdom}
\EventLogo{}
\SeriesVolume{42}
\ArticleNo{23}
\begin{document}
\nolinenumbers
\maketitle

\begin{abstract}
The fundamental model of a periodic structure is a periodic point set up to rigid motion or isometry. Our recent paper in SoCG 2021 defined isometry invariants (density functions), which are complete in general position and continuous under perturbations. This work introduces much faster isometry invariants (average minimum distances), which are also continuous and distinguish some sets that have identical density functions. We explicitly describe the asymptotic behaviour of the new invariants for a wide class of sets including non-periodic. The proposed near linear time algorithm processed a dataset of hundreds of thousands of real structures in a few hours on a modest desktop.
\end{abstract}

%1========================
\section{Motivations, problem statement and overview of new results}
\label{sec:intro}

A \emph{lattice} consists of integer linear combinations of a basis whose vectors span a parallelepiped called a \emph{unit cell}.
A \emph{periodic point set} is the Minkowski sum $\La+M=\{\vec u+\vec v \vl u\in\La, v\in M\}$ of a \emph{lattice} $\La$ and a \emph{motif} $M$, which is a finite set of points in the unit cell $U$ \cite{zhilinskii2016introduction}.
\smallskip

The recent papers \cite{edels2021,anosova2021isometry} of the fourth author initiated a study into efficient classifications of periodic point sets up to isometry.
An \emph{isometry} of Euclidean $\R^n$ is any map that preserves inter-point distances.
Any orientation-preserving isometry can be realised as a continuous rigid motion, for example a composition of translations and rotations in $\R^3$.
This equivalence is most natural for periodic point sets due to their applications in rigid periodic structures.  
\smallskip

An isometry classification of periodic point sets turned out to be highly non-trivial from both mathematical and algorithmic points of view for the  following four reasons.
%\myskip

\vspace*{-1mm}
\begin{figure*}[h!]
\includegraphics[height=19mm]{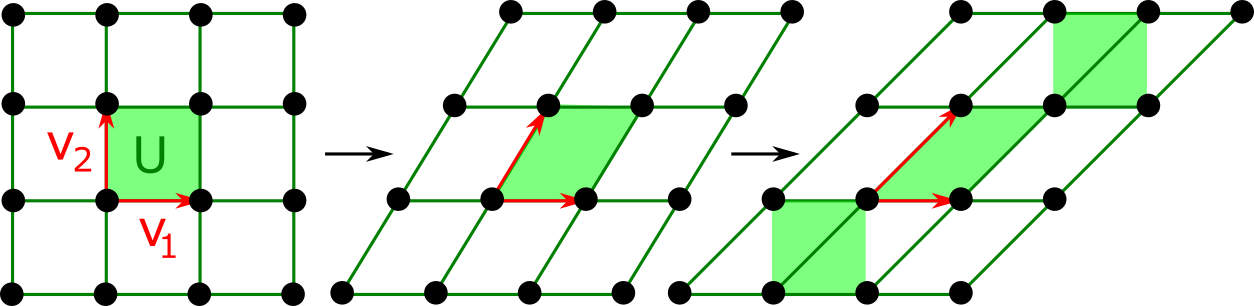}
\hspace*{1mm}
\includegraphics[height=19mm]{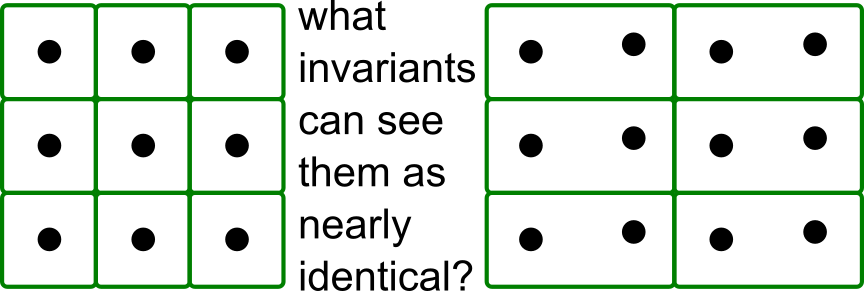}
\caption{\textbf{Left}: a continuous deformation gives different bases for the same first and last lattices.
\textbf{Right}: many isometry invariants such as symmetry groups are discontinuous under perturbations.}
\label{fig:deformations}
\end{figure*}

Firstly, since any lattice can be generated by infinitely many different bases or \emph{primitive} unit cells of a minimal volume, a representation of a periodic point set as a sum $S=\La+M$ of a lattice and a motif is highly ambiguous.
All attempts to find a canonical basis or a \emph{reduced} cell of a lattice are discontinuous under perturbations as noted in \cite[section~1]{edels2021}.
%, see Theorem~\ref{thm:discontinuity} in section~\ref{sec:discussion}.
\medskip

Fig.~\ref{fig:deformations} shows that all periodic point sets form a continuous space whose geometry remains unexplored \cite{cassels2012introduction}. 
The last two periodic sets in Fig.~\ref{fig:deformations} provide an example of nearly identical structures whose similarity is hard to quantify without using thresholds.
These sets substantially differ by symmetry groups and unit cell volumes but have the same density.
\medskip

Secondly, even for a fixed lattice basis, shifting points within a unit cell changes their coordinates in the basis.
Similarly, any isometry of a periodic point set preserves a rigid structure, but dramatically changes the basis representation in a fixed coordinate system.
\medskip

Thirdly, crystallography studies periodic structures via symmetry groups, which break down under perturbations of points.
Many past invariants are discrete and cut the continuous space of periodic point sets into strata.
This discontinuous stratification is a main obstacle for understanding transitions between crystal phases and for detecting nearly identical structures, which are common due to noise in measurements or approximations in simulations. 
\medskip

Fourthly, the world's largest Cambridge Structural Database already contains more than 1.1M known structures kept as pairs (unit cell, motif), which should be converted into continuous invariants for reliable comparisons.
Hence we need quickly computable invariants that can be sequentially extended to distinguish more and more complicated structures.
\medskip

This paper contributes to both mathematics and computer science by solving the following problem to find a sequential family of isometry  invariants computable in near linear time.

\begin{pro}[fast and continuous isometry invariants]
\label{pro:invariants}
Find continuous isometry invariants of a periodic point set that can be computed in a near linear time in all input parameters. 
\bs
\end{pro}

\begin{figure*}[h!]
\includegraphics[width=\linewidth]{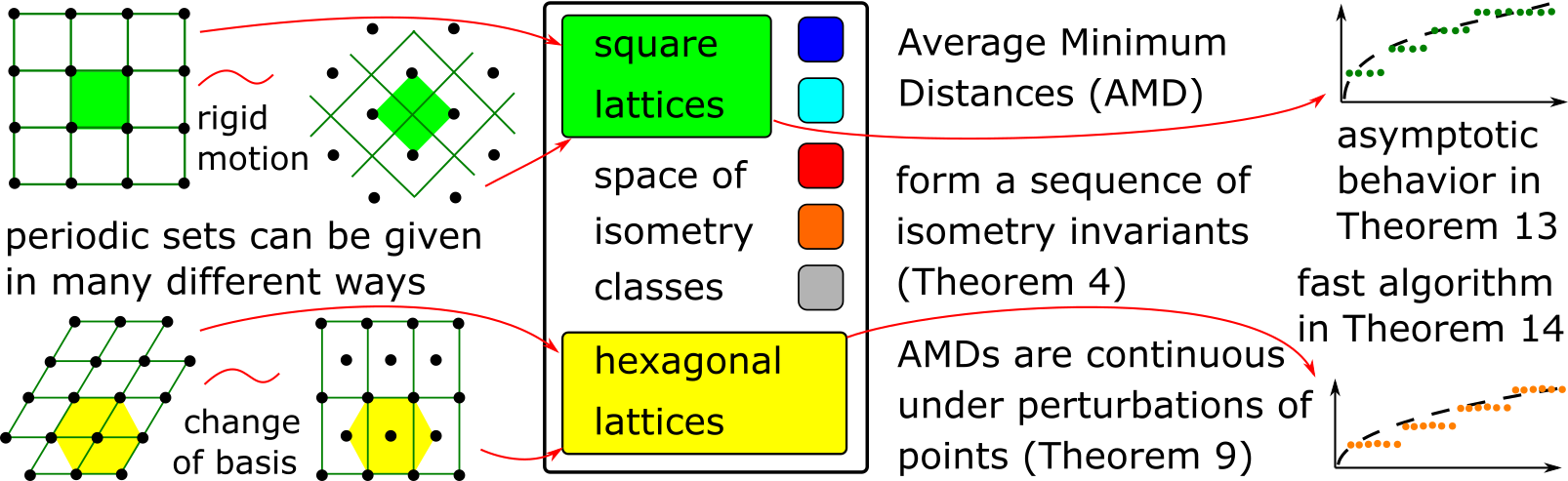}
\caption{Ambiguous representations of periodic point sets by pairs (unit cell, motif) are converted into continuous isometry invariants (Average Minimum Distances), which form an infinite sequence.}
\label{fig:invariants}
\end{figure*}

Fig.~\ref{fig:invariants} summarises the problem of distinguishing periodic point sets up to isometry and outlines  properties of the new Average Minimum Distances (AMD).
Section~\ref{sec:review} reviews key concepts and past results on isometry classifications.
Section~\ref{sec:AMD_invariant} introduces the AMD and illustrates computations by examples.
Section~\ref{sec:continuity} proves the Lipschitz continuity of AMD under perturbations of points.
Section~\ref{sec:asymptotic} proves the asymptotic behaviour of the infinite AMD sequence.
Section~\ref{sec:discussion} describes a near linear time algorithm to compute AMD, which runs  on real structures in milliseconds on a modest desktop.
The paper finishes with a discussion of current limitations and potential developments.
The appendices for crystallographers include extra details, comparisons with past applied work and breakthrough visualisations.

%2====================
\section{Key definitions for periodic point sets and a review of past work}
\label{sec:review}

The symbol $\R^n$ denotes Euclidean space. 
Any point $p\in\R^n$ can be represented by the vector $\vec p$ from the origin of $\R^n$ to the point $p$, so $p$ and $\vec p$ can be used interchangeably, though $\vec p$ can be drawn at any initial point.
The \emph{Euclidean} distance between points $p,q\in\R^n$ is $|p-q|$.

\begin{dfn}[a lattice $\La$, a motif $M$, a unit cell $U$, a periodic point set $S=\La+M$]
\label{dfn:crystal}
Let vectors $\vec v_1,\dots,\vec v_n$ form a linear {\em basis} in $\R^n$ so that if $\sum\limits_{i=1}^n \la_i\vec v_i=\vec 0$ for some real $\la_i$, then all $\la_i=0$.
Then a {\em lattice} $\La$ in $\R^n$ consists of all linear combinations $\sum\limits_{i=1}^n \la_i\vec v_i$  with integer coefficients $\la_i\in\Z$.
A {\em motif} $M$ is a finite set of points $p_1,\dots,p_m$ in the \emph{unit cell} $U(\vec v_1,\dots,\vec v_n)=\left\{ \sum\limits_{i=1}^n \la_i\vec v_i \vl \la_i\in[0,1) \right\}$, which is the parallepiped spanned by $\vec v_1,\dots,\vec v_n$.
\medskip

A \emph{periodic point set} $S\subset\R^n$ is the \emph{Minkowski sum}  $S=\La+M=\{\vec u+\vec v \vl  u\in\La, v\in M\}$, so $S$ is a finite union of translates of the lattice $\La$.
A unit cell $U$ is \emph{primitive} if $S$ remains invariant under shifts by vectors only from $\La$ generated by $U$ (or the basis $\vec v_1,\dots,\vec v_n$).
\bs
\end{dfn}

Any lattice $\La$ can be considered as a periodic set with a 1-point motif $M=\{p\}$.
This single point $p$ can be arbitrarily chosen in a unit cell $U$.
The lattice translate $\La+\vec p$ is also considered as a lattice, because the point $p$ can be chosen as the origin of $\R^n$.
\medskip

The periodic sets in the top left part of Fig.~\ref{fig:invariants} represent isometric square lattices, though the former has a point $p$ at a corner of a unit cell $U$ and the latter has $p$ in the center of $U$.
The periodic sets in the bottom left part of Fig.~\ref{fig:invariants} represent isometric  hexagonal lattices, because every black point has exactly six nearest neighbours that form a regular hexagon.
\medskip

A lattice $\La$ of a periodic set $S=M+\La\subset\R^n$ is not unique in the sense that $S$ can be generated by a sublattice of $\La$ and a motif larger than $M$.
If $U$ is any unit cell of $\La$, the sublattice $2\La$ has the $2^n$ times larger unit cell $2^n U$ (twice larger along each of $n$ basis vectors of $U$), hence contains $2^n$ times more points than $M$.
Such an extended cell $2^n U$ is superfluous, because $S$ is invariant under translations along not only integer linear combinations $\sum\limits_{i=1}^n \la_i\vec v_i$ with $\la_i\in\Z$, but also along linear  combinations with half-integer coefficients $\la_i\in \frac{1}{2}\Z$.
\medskip

%\section{Related work on comparisons of point sets and crystals}
%\label{sec:review}

Now we discuss the closely related work on comparing finite and periodic sets up to isometry.  
The excellent book \cite{liberti2017euclidean} reviews the wider area of Euclidean distance geometry.
\medskip

The full distribution of all pairwise Euclidean distances $|a-b|$ between points $a,b$ in a finite set $S\subset\R^m$ is a well-known isometry invariant.
This invariant is complete or injective for finite sets in general position \cite{boutin2004reconstructing} in the sense that almost any finite set $S$ can be uniquely reconstructed up to isometry from the set of all distances between points of $S$.
The left hand side pictures in Fig.~\ref{fig:non-isometric_pairs} show the counter-example pair $T,K$ to the full completeness. % of the distance distribution.
\medskip

\begin{figure}[h]
\includegraphics[height=26mm]{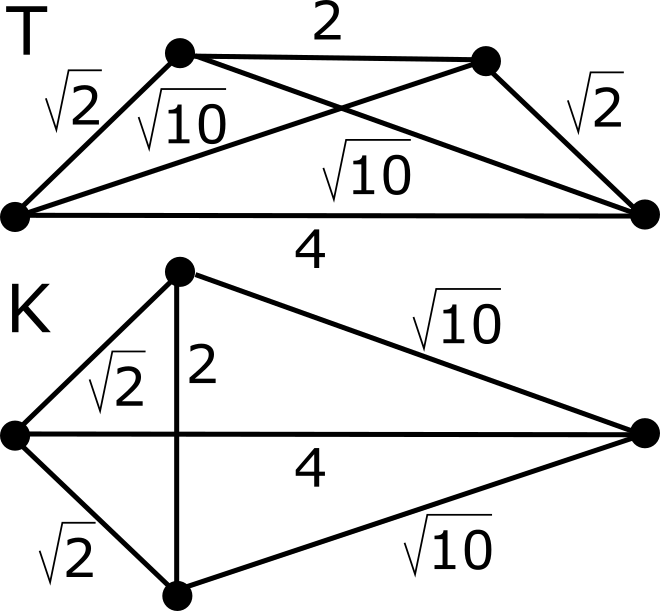}
\hspace*{0mm}
\includegraphics[height=26mm]{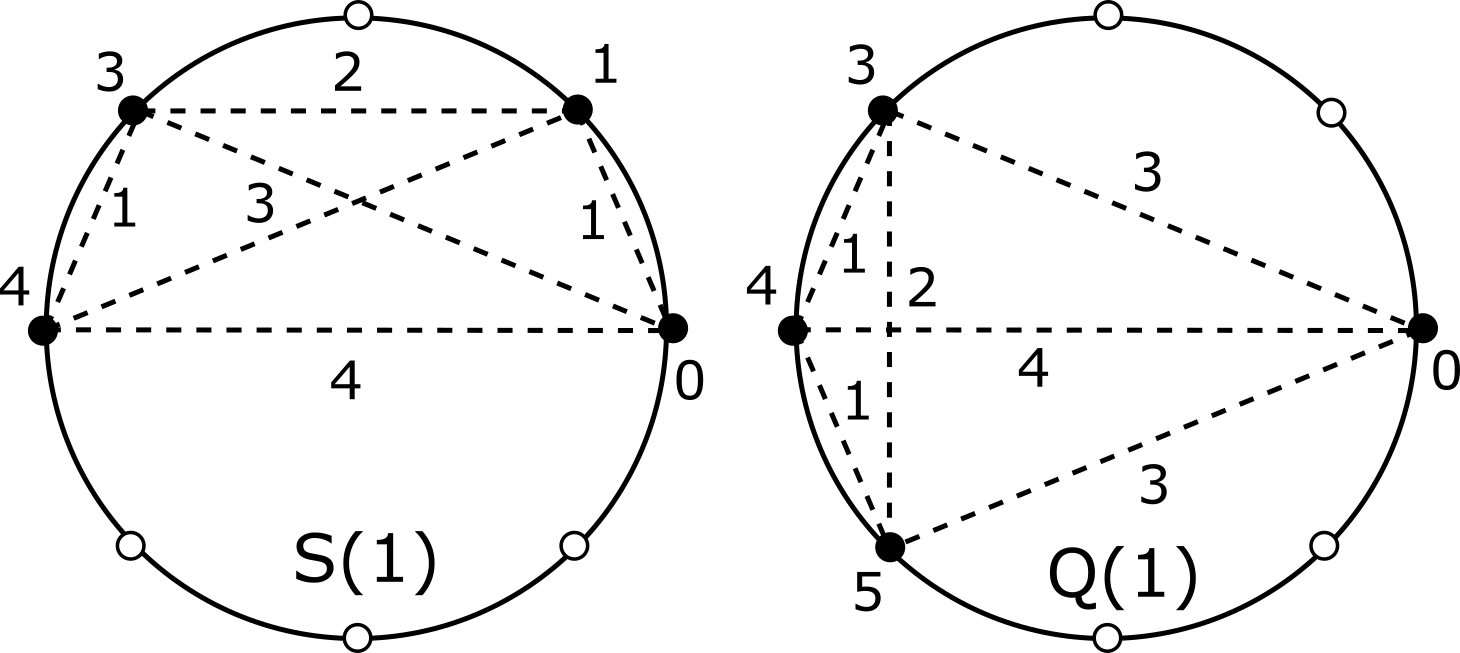}
\hspace*{1mm}
\includegraphics[height=26mm]{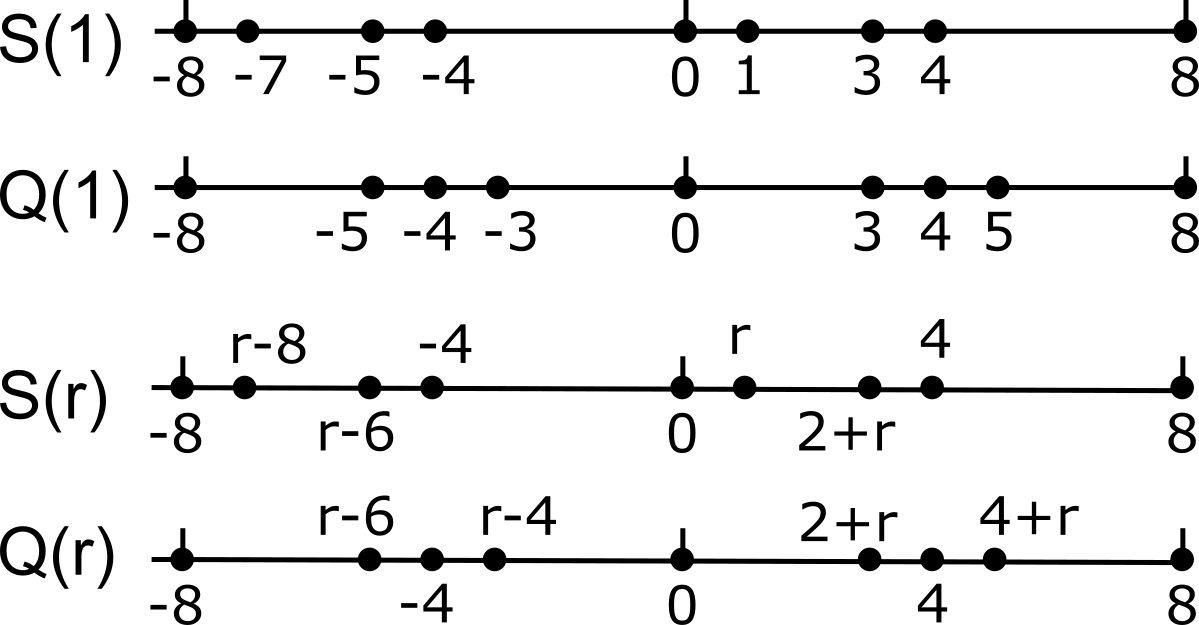}
\caption{Non-isometric sets that cannot be distinguished by past invariants. 
\textbf{Left}: $T,K\subset\R^2$ have the same pairwise distances $\{2,\sqrt{2},\sqrt{2},\sqrt{10},\sqrt{10},4\}$.
\textbf{Middle}: periodic versions $S(1),Q(1)$ of $T,K$ on a circle of length 8. Distances are measured along arcs.
\textbf{Right}: periodic sets $S(r)=\{0, r, r+2, 4\}+8\Z$ and $Q(r)=\{0, r+2, 4, r+4\}+8\Z$ for $0<r\leq 1$ have the same Patterson function \cite[p.~197, Fig.~2]{patterson1944ambiguities}.
All these pairs are distinguished by $\AMD$ invariants in Section~\ref{sec:AMD_invariant}.}
\label{fig:non-isometric_pairs}
\end{figure}

The isometry classification of finite point sets was algorithmically resolved by \cite[Theorem~1]{alt1988congruence} saying that an existence of an isometry between $m$-point sets in $\R^n$ can be checked in time $O(m^{n-2}\log m)$.
The latest advance is the sophisticated $O(m\log m)$ algorithm \cite{kim2016congruence} testing an isometry (or congruence) between $m$-point sets in $\R^4$.
For finite sets, the new concept of Average Minimum Distances is similar to M\'emoli's seminal work on \emph{distributions of distances }\cite{memoli2011gromov}, also known as \emph{shape distributions} \cite{belongie2002shape, grigorescu2003distance, manay2006integral, osada2002shape}.
All algorithms for finite sets cannot be easily extended to periodic sets by   choosing a unit cell, because any such reduced cell is discontinuous under perturbations according to Theorem~\ref{thm:discontinuity} in Section~\ref{sec:discussion}. 
\medskip

An isometry classification of periodic point sets by complete invariants is already non-trivial in dimension~1. 
Gr\"unbaum and Moore \cite{grunbaum1995use} constructed the sophisticated invariants below.
Points $c_0,\dots,c_{m-1}\in\R$ of a 1-dimensional periodic set can be considered on the unit circle $S^1\subset\mathbb{C}$ and converted into the Fourier coefficients $d(k)=\sum\limits_{j=0}^{m-1} c_j\exp\dfrac{2\pi\sqrt{-1}jk}{m}$, $k=0,\dots,m-1$. 
Then all 1-dimensional periodic sets whose points have only integer (or rational) coordinates are distinguished up to translations by the $n$-th order invariants up to $n=6$, which are all products of the form $d(k_1)\cdots d(k_n)$ with $k_1+\dots+k_n\equiv 0\pmod{m}$.
\medskip

The 3-dimensional analogues of the above invariant products $d(k)d(-k)$ define the \emph{Patterson function}, whose peaks correspond to interpoint vectors \cite{glusker1987patterson}.
Periodic sets that have identical Patterson functions are called \emph{homometric}, see more details in appendix~B.
\medskip

The 4-point non-isometric sets $T,K$ have periodic versions 
$S(1)=\{0, 1, 3, 4\}+8\Z$ and $Q(1)=\{0, 3, 4, 5\}+8\Z$ in  Fig.~\ref{fig:non-isometric_pairs}.
Even more general homometric sets $S(r),Q(r)$ depending on a parameter $0<r\leq 1$ will be distinguished by simplest invariant $\AMD_1$ in section~\ref{sec:AMD_invariant}.
\medskip

More recently, for any periodic point set $S\subset\R^n$ with a motif $M$ in a unit cell $U$, Edelsbrunner et al. \cite{edels2021} introduced the density functions $\psi_k(t)$ for any integer $k\geq 1$.
The $k$-th \emph{density function} $\psi_k(t)$ is the total volume of the regions within the unit cell $U$ covered by exactly $k$ balls $B(p;t)$ with a radius $t\geq 0$ and centres at points $p\in M$, divided by the unit cell volume $\vol[U]$.
The density function $\psi_k(t)$ was proved to be invariant under isometry, continuous under perturbations, complete for periodic sets satisfying certain conditions of general position in $\R^3$, and computable in time $O(m k^3)$, where $m$ is the motif size of $S$.
\medskip

Section~5 in \cite{edels2021} gives the counter-example to completeness: the 1-dimensional periodic sets $S_{15}=X+Y+15\Z$ and $Q_{15}=X-Y+15\Z$ for $X = \{0, 4, 9\}$ and $Y = \{0, 1, 3\}$, which appeared earlier in \cite[section~4]{grunbaum1995use}.
These non-isometric sets have the same density functions for all $k\geq 1$, see \cite[Example~11]{anosova2021introduction}, and will be distinguished by $\AMD_3$ in Example~(\ref{exa:SQr+SQ15}b).
% in section~\ref{sec:AMD_invariant}.
\medskip

The latest advance \cite{anosova2021isometry} reduces the isometry classification of all periodic point sets to an \emph{isoset} of isometry classes of $\al$-clusters around points in a motif at a certain radius $\al$, which was motivated by the seminal work of Dolbilin with co-authors about Delone sets \cite{dolbilin1998multiregular,bouniaev2017regular,dolbilin2019regular}.
Checking if two isosets coincide needs a cubic algorithm, which is not yet implemented.
Running times of available algorithms will be compared in section~\ref{sec:discussion}.  
Appendix~C reviews past experimental work on comparisons of periodic structures for experts in applications.
\medskip

The only other invariant that is widely used in applications which is continuous under perturbations is the density $\rho$ equal to the weight of atoms within a unit cell $U$, divided by $\vol[U]$.
This density $\rho$ does not distinguish any perturbed periodic point sets, see Fig.~\ref{fig:deformations}.

%3============
\section{Definition and isometry invariance of Average Minimum Distances}
\label{sec:AMD_invariant}

This section introduces Average Minimum Distances in Definition~\ref{dfn:AMD} and proves their isometry invariance in Theorem~\ref{thm:AMD_invariant}.
For the particular case of a lattice $\La\subset\R^n$, $\AMD_k(\La)$ can be defined as the distance from a fixed point $p\in L$ to its $k$-th nearest neighbour in $\La$, see Fig.~\ref{fig:AMD}.

\newcommand{\amdh}{45mm}
\begin{figure*}[h]
\centering
\includegraphics[height=\amdh]{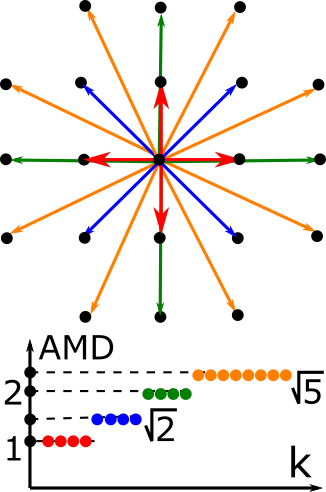}
\hspace*{8mm}
\includegraphics[height=\amdh]{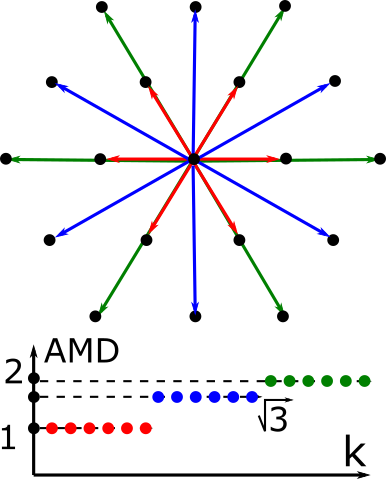}
\hspace*{8mm}
\includegraphics[height=\amdh]{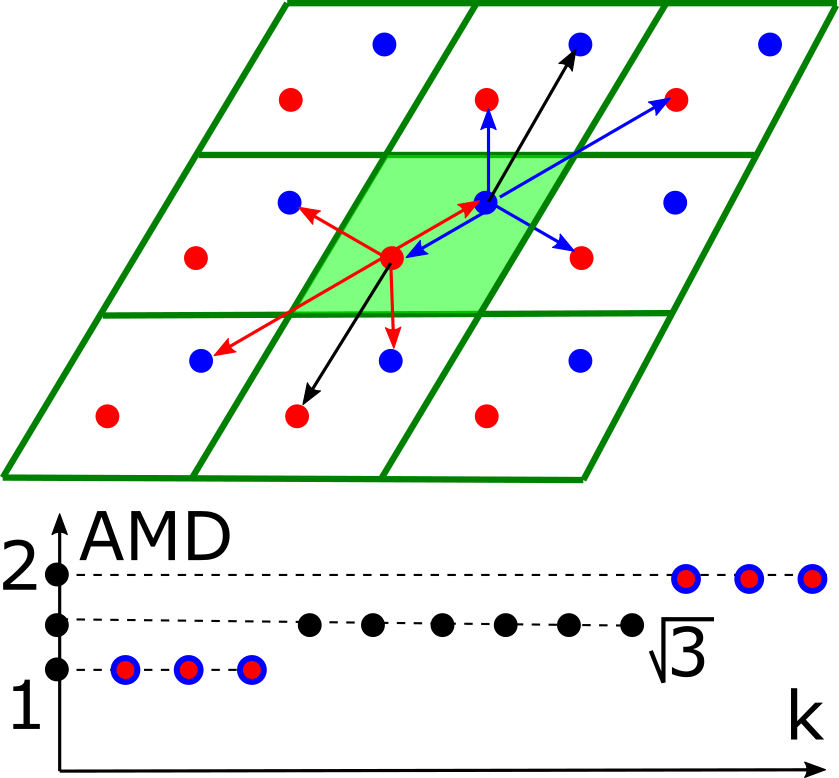}
\caption{\textbf{Left}: in the square lattice, the $k$-th neighbours of the origin and corresponding $\AMD_k$  are shown in the same colour, e.g. the shortest axis-aligned distances $\AMD_1=\dots=\AMD_4=1$ are in red, the longer diagonal distances $\AMD_5=\dots=\AMD_8=\sqrt{2}$ are in blue.
\textbf{Middle}: in the hexagonal lattice, the shortest distances are in red: $\AMD_1=\dots=\AMD_6=1$.
\textbf{Right}: the motif of two points (red and blue) within a green unit cell has 
%the the first three averages 
$\AMD_1=\AMD_2=\AMD_3=1$.}
\label{fig:AMD}
\end{figure*}

\begin{dfn}[Average Minimum Distances $\AMD_k(S)$]
\label{dfn:AMD}
Let a periodic point set $S=\La+M\subset\R^n$ have points $p_1,\dots,p_m$ in a primitive unit cell.
For a fixed integer $k\geq 1$ and $i=1,\dots,m$, the $i$-th row of the $m\times k$ matrix $D(S;k)$ consists of the ordered Euclidean distances $d_{i1}\leq\cdots\leq d_{ik}$ measured from the point $p_i$ to its first $k$ nearest neighbours within the infinite set $S$, see Fig.~\ref{fig:AMD}.
Then the \emph{Average Minimum Distance}
$\AMD_k(S)=\dfrac{1}{m}\sum\limits_{i=1}^m d_{ik}$ equals the average of the $k$-th column in the matrix $D(S;k)$ of distances to neighbours.
\bs
\end{dfn}

Definition~\ref{dfn:AMD} makes sense for any finite set $S=M$ of $m$ points for $k\leq m-1$.
Then the matrix $D(S;m-1)$ for the largest possible number $k=m-1$ of neighbours includes all pairwise distances, but differs from the usual symmetric distance matrix of $S$ due to the ordered distances in each row.
This pointwise information distinguishes the 4-point sets $T,K$ in Fig.~\ref{fig:non-isometric_pairs} as follows.
The trapezium $T$ and kite $K$ in $\R^2$ can be represented by the points $(\pm 1,1),(\pm 2, 0)$ and $(-2,0),(-1, \pm 1),(2,0)$, respectively.
The matrices from Definition~\ref{dfn:AMD} are
$D(T;3)=
\left(\begin{array}{lll}
\sqrt{2} & 2 & \sqrt{10} \\
\sqrt{2} & 2 & \sqrt{10} \\
\sqrt{2} & \sqrt{10} & 4 \\
\sqrt{2} & \sqrt{10} & 4 
\end{array}\right)$ and
$D(K;3)=
\left(\begin{array}{lll}
\sqrt{2} & \sqrt{2} & \sqrt{10} \\
\sqrt{2} & 2 & \sqrt{10} \\
\sqrt{2} & 2 & \sqrt{10} \\
\sqrt{10} & \sqrt{10} & 4 
\end{array}\right)$.
Already the first components of the vectors
$\AMD^{(3)}(T)=(\sqrt{2},1+\frac{\sqrt{10}}{2}, 2+\frac{\sqrt{10}}{2})$ and $\AMD^{(3)}(K)=(\frac{3\sqrt{2}+\sqrt{10}}{4},1+\frac{\sqrt{2}+\sqrt{10}}{4}, 1+\frac{3}{4}\sqrt{10})$ distinguish the sets $K,T$ in Fig.~\ref{fig:non-isometric_pairs}.
\medskip

If $S$ is periodic, all AMD values form the infinite sequence $\{\AMD_k\}_{k=1}^{+\infty}$.
In practice, we compute the vector $\AMD^{(k)}=(\AMD_1,\dots,\AMD_k)$ up to a number $k$ of nearest neighbours.
However, $k$ is not a parameter that changes the output.
If we increase $k$, we get more values without changing the previous ones, so $k$ is similar to a number of decimal places or a length of approximation.
Since the asymptotic behaviour of $\AMD_k$ will be explicitly described in Theorem~\ref{thm:asymptotic}, the infinite AMD sequence can be informally compared with the sequence of coefficients in a degree $k$ Taylor polynomial approximating an analytic function as $k\to+\infty$.
%\medskip
 
\begin{theorem}[isometry invariance of AMD]
\label{thm:AMD_invariant}
For any finite or periodic point set $S\subset\R^n$, the Average Minimum Distance $\AMD_k(S)$ 
%from Definition~\ref{dfn:AMD} 
is an isometry invariant of $S$ for any $k\geq 1$.
\bs
\end{theorem}
\begin{proof}[Proof of Theorem~\ref{thm:AMD_invariant}]
If $S$ is periodic, first we show that the unordered collections of rows of the matrix $D(S;k)$, and hence $\AMD_k(S)$, is independent of a primitive unit cell. 
Let $U,U'$ be different primitive cells of the periodic point set $S\subset\R^n$ with a lattice $\La$.
Any point $q\in S\cap U'$ can be translated along $\vec v\in\La$ to a point $p\in S\cap U$ and vice versa.
These translations establish a bijection between the motifs $S\cap U\lra S\cap U'$ and preserve all distances.
Hence the matrix $D(S;k)$ is the same for both cells $U,U'$ up to a permutation of rows.
\medskip

Now we prove that $D(S;k)$, and hence $\AMD_k(S)$, is preserved by
any isometry $f:S\to Q$.
Any primitive cell $U$ of $S$ is bijectively mapped by $f$ to the unit cell $f(U)$ of $Q$, which should be also primitive.
Indeed, if $Q$ is preserved by a translation along a vector $\vec v$ that  doesn't have all integer coefficients in the basis of $f(U)$, then
$S=f^{-1}(Q)$ is preserved by the translation along $f^{-1}(\vec v)$, which also doesn't have all integer coefficients in the basis of $U$, i.e. $U$ was non-primitive.
Since both primitive cells $U$ and $f(U)$ contain the same number of points from $S$ and $Q=f(S)$, the isometry $f$ gives a bijection between all motif points of $S,Q$.
\medskip

For any (finite or periodic) sets $S,Q$, since $f$ preserves distances, every list of ordered distances from any point $p_i\in S\cap U$ to its first $k$ nearest neighbours in $S$ coincides with the list of the ordered distances from $f(p_i)$ to its first $k$ neighbours in $Q$.
The matrices $D(S;k),D(Q;k)$ are identical up to permutations of rows, hence $\AMD_k(S)=\AMD_k(Q)$.  
\end{proof}

\begin{exa}
\label{exa:SQr+SQ15}
\textbf{(\ref{exa:SQr+SQ15}a)}
Table~\ref{tab:SQr} implies by Theorem~\ref{thm:AMD_invariant} that $S(r),Q(r)$ in Fig.~\ref{fig:SQr+SQ15} are not isometric for $0<r\leq 1$.
%Fig.~\ref{fig:SQr+SQ15} shows the periodic sets $S(r),Q(r)$. 
The mirror image of $S(r)=\{0, r, r+2, 4\}+8\Z$ under the reflection $t\mapsto 4-t$ coincides with $S(2-r)=\{0, 2-r, 4-r, 4\}+8\Z$, so they are equivalent up to isometries including reflections.
Similarly, $Q(r)$ and $Q(2-r)$ are isometric by $t\mapsto -t$.
%To distinguish $S(r),Q(r)$, we can assume that the parameter $r$ satisfies $0<r\leq 1$.
Though $\AMD_k(S(r))$ seem to be independent of $r$, the first column of $D(S(r);3)$ has the minimum distance $r$, which distinguishes $S(r)$ between each other for different parameters $0<r\leq 1$.
\medskip

\noindent
\textbf{(\ref{exa:SQr+SQ15}b)}
The sets 
$S_{15} = \{0,1,3,4,5,7,9,10,12\}+15\Z$,
$Q_{15} = \{0,1,3,4,6,8,9,12,14\}+15\Z$ in Fig.~\ref{fig:SQr+SQ15}
 are not isometric, because $\AMD_k(S_{15})\neq\AMD_k(Q_{15})$ for $k=3,4$ in Table~\ref{tab:SQ15}.
All their density functions $\psi_k(t)$ are identical as noted in the beginning of \cite[section~5]{edels2021}.
\bs
\end{exa}

\begin{table*}[h!]
$\begin{array}{l|lllll}
D(S(r);3) & \mbox{distance to 1st neighb.} & \mbox{distance to 2nd neighb.} & \mbox{distance to 3rd neighb.}  %&  \mbox{4th distance} &  \mbox{5th distance} 
\\
\hline
p_1=0 & |0-r|=r & |0-(2+r)|=2+r & |0-4|=4 %& |0-(4-8)|=4 & |0-(2+r-8)|=6-r 
\\
p_2=r & |r-0|=r & |r-(2+r)|=2 & |r-4|=4-r %& |r-(4-8)|=4+r & |r-(2+r-8)|=6
 \\
p_3=2+r & |(2+r)-4|=2-r & |(2+r)-r|=2 & |(2+r)-0|=2+r %& |(2+r)-8|=6-r & |(2+r)-(r+8)|=6 
\\
p_4=4 & |4-(2+r)|=2-r & |4-r|=4-r & |4-0|=4 %& |4-8|=4 & |4-(r+8)|=4+r 
\\
\hline
\AMD_k(S(r)) & \AMD_1=1 & \AMD_2=2.5 & \AMD_3=3.5 %& \AMD_4=4.5 & \AMD_5=5.5
\end{array}$
\medskip

$\begin{array}{l|lllll}
D(Q(r);3) & \mbox{distance to 1st neighb.} & \mbox{distance to 2nd neighb.} & \mbox{distance to 3rd neighb.} %&  \mbox{4th distance} &  \mbox{5th distance} 
\\
\hline
p_1=0 & |0-(2+r)|=2+r & |0-(r+4-8)|=4-r & |0-4|=4 %& |0-(4-8)|=4 & |0-(r+4)|=r+4 
\\
p_2=2+r & |(2+r)-4|=2-r & |(2+r)-(4+r)|=2 & |(2+r)-0|=2+r % & |2+r-8|=6-r & |2+r-(r-4)|=6 
\\
p_3=4 & |4-(4+r)|=r & |4-(2+r)|=2-r & |4-0|=4 %& |4-8|=4 & |4-(10+r)|=6+r 
\\
p_4=4+r & |(4+r)-4|=r & |(4+r)-(2+r)|=2 & |(4+r)-8|=4-r % & |(4+r)-0|=4+r & |(4+r)-(10+r)|=6 
\\
\hline
\AMD_k(Q(r)) & \AMD_1=1+0.5r & \AMD_2=2.5-0.5r & \AMD_3=3.5 %& \AMD_4=4.5 & \AMD_5=5+0.5r
\end{array}$
\medskip

\caption{Matrices $D(S;k)$ and average minimum distances (AMD) from Definition~\ref{dfn:AMD} distinguish the periodic sets $S(r)=\{0, r, r+2, 4\}+8\Z$ and $Q(r)=\{0, r+2, 4, r+4\}+8\Z$ for any $0<r\leq 1$.}
\label{tab:SQr}
\end{table*}

\newcommand{\hheight}{30mm}
\begin{figure}[h]
\centering
\includegraphics[height=\hheight]{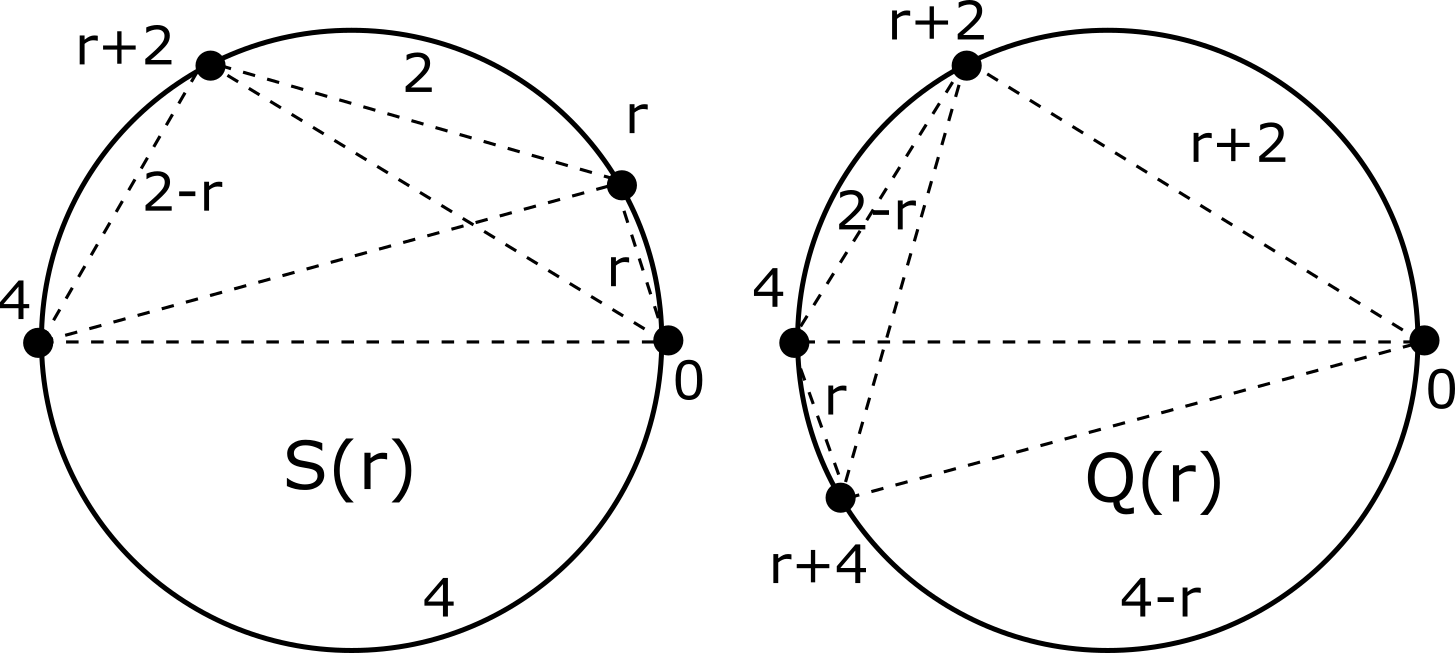}
\hspace*{2mm}
\includegraphics[height=\hheight]{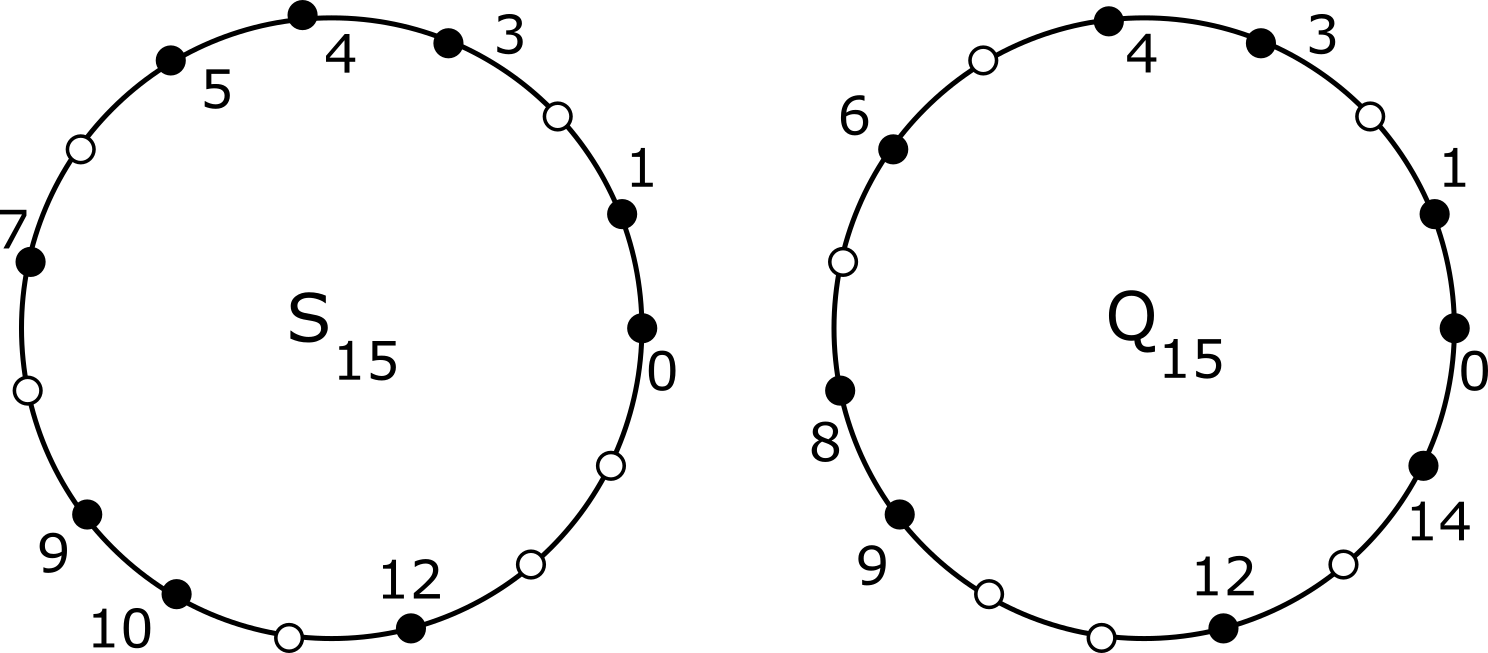}
\caption{\textbf{Left}: circular versions of the periodic point sets $S(r)=\{0, r, r+2, 4\}+8\Z$ and $Q(r)=\{0, r+2, 4, r+4\}+8\Z$ for $0<r\leq 1$.
The distances between points (shown outside the disk) are arc lengths (shown inside the disk).
\textbf{Right}: the periodic sets $S_{15},Q_{15}$ from Example~\ref{exa:SQr+SQ15}b are distinguished by $\AMD_3$, not by the density functions $\psi_k(t)$ for all $k\geq 1$ as shown in \cite[Example~11]{anosova2021introduction}. }
\label{fig:SQr+SQ15}
\end{figure}

\vspace*{-4mm}
  
\noindent
\begin{table}[h!]
\begin{tabular}{C{8mm}|C{0.8mm}C{0.8mm}C{0.8mm}C{0.8mm}C{0.8mm}C{0.8mm}C{0.8mm}C{1mm}C{2mm}|C{8mm}|}
\hline\noalign{\smallskip}
$S_{15}$ & 0 & 1 & 3 & 4 & 5 & 7 & 9 & 10 & 12 & $\AMD_k$ \\
\hline
%\noalign{\smallskip}\noalign{\smallskip}
$k=1$      & 1 & 1 & 1 & 1 & 1 & 2 & 1 &  1 &  2 & 11/9 \\
%\noalign{\smallskip}\noalign{\smallskip}
$k=2$      & 3 & 2 & 2 & 1 & 2 & 2 & 2 &  2 &  3 & 19/9 \\
%\noalign{\smallskip}\noalign{\smallskip}
$k=3$      & 3 & 3 & 2 & 3 & 2 & 3 & 3 &  3 &  3 & 25/9 \\
%\noalign{\smallskip}\noalign{\smallskip}
$k=4$      & 4 & 4 & 3 & 3 & 4 & 3 & 4 &  5 &  4 & 34/9 
\end{tabular}
\begin{tabular}{|C{8mm}|C{0.8mm}C{0.8mm}C{0.8mm}C{0.8mm}C{0.8mm}C{0.8mm}C{0.8mm}C{1mm}C{2mm}|C{6mm}}
\hline\noalign{\smallskip}
$Q_{15}$ & 0 & 1 & 3 & 4 & 6 & 8 & 9 & 12 & 14 & $\AMD_k$ \\
\hline
%\noalign{\smallskip}\noalign{\smallskip}
$k=1$      & 1 & 1 & 1 & 1 & 2 & 1 & 1 &  2  & 1 & 11/9 \\
%\noalign{\smallskip}\noalign{\smallskip}
$k=2$      & 1 & 2 & 2 & 2 & 2 & 2 & 3 &  3  & 2 & 19/9 \\
%\noalign{\smallskip}\noalign{\smallskip}
$k=3$      & 3 & 2 & 3 & 3 & 3 & 4 & 3 &  3  & 2 & 26/9 \\
%\noalign{\smallskip}\noalign{\smallskip}
$k=4$      & 3 & 3 & 3 & 4 & 3 & 4 & 5 &  4  & 4 & 33/9
\end{tabular}
\caption{\textbf{First row}: 9 points from the motif $M$ of the periodic sets $S_{15}$ (left) and $Q_{15}$ (right) in Fig.~\ref{fig:SQr+SQ15}.
\textbf{Further rows}: distances from each $p\in M$ to its $k$-th nearest neighbour in $S_{15}, Q_{15}$.}
\label{tab:SQ15}
\end{table}
  
%4============
\section{Continuity of Average Minimum Distances under perturbations}
\label{sec:continuity}

For the isometry invariance of $\AMD(S;k)$ in Theorem~\ref{thm:AMD_invariant}, a unit cell $U$ in Definition~\ref{dfn:AMD} should be primitive.
If $U$ contains $m$ points and we make one edge of $U$ twice longer, the resulting non-primitive unit cell contains $2m$ points and the matrix $D(S;k)$ will be twice larger.
A translated copy of any point $p_i\in U$ will have exactly the same ordered distances to its neighbours as $p_i$ due to periodicity.
After doubling $U$, every row is repeated twice in $D(S;k)$.
The requirement of a primitive cell $U$ makes $D(S;k)$ discontinuous similarly to the cell volume $\vol[U]$ in Fig.~\ref{fig:deformations}.
One way to resolve this discontinuity is to average each column.
\medskip

The bottleneck distance $d_B$ below measures a maximum perturbation of points needed to get one set from another.
Such perturbations always exist due to thermal vibrations.

\begin{dfn}[bottleneck distance]
\label{dfn:BND}
For a bijection $g:S\to Q$ between finite or periodic sets $S,Q\subset\R^n$, the \emph{maximum deviation} is the supremum $\sup\limits_{p\in S}|p-g(p)|$ 
%of Euclidean distances 
over $p\in S$.
The \emph{bottleneck distance} is $d_B(S,Q)=\inf\limits_{g:S\to Q}\;\sup\limits_{p\in S}|p-g(p)|$ is the infimum over bijections $g:S\to Q$.
\bs
\end{dfn}

The bottleneck distance is impractical to compute because of a minimisation over infinitely many bijections.
Theorem~\ref{thm:discontinuity} in Section~\ref{sec:discussion} will justify that there is no continuous way to select a unit cell of a lattice. 
However, continuity of isometry invariants can be proved for perturbations in the bottleneck distance.
Continuity Theorem~\ref{thm:continuity} requires Lemmas~\ref{lem:common_lattice} and \ref{lem:perturbed_distances}.

\begin{lem}[Lemma~2 in \cite{edels2021}]
\label{lem:common_lattice}
Let periodic point sets $S,Q\subset\R^n$ have a bottleneck distance $d_B(S,Q)<r(Q)$, where $r(Q)$ is the minimum half-distance between points of $S$.
%Alternatively, $r(S)$ is the maximum radius $r$ to have disjoint open balls $B(p;r)$ of radius $r$ centered at all points $p\in S$.
Then $S,Q$ have a common lattice $\La$ with a unit cell $U$ such that $S=\La+(U\cap S)$, $Q=\La+(U\cap Q)$.
\bs
\end{lem}

\begin{lem}[perturbed distances]
\label{lem:perturbed_distances}
For some $\ep>0$, let $g:S\to Q$ be a bijection between finite or periodic sets such that $|a-g(a)|\leq\ep$ for all $a\in S$.
For any $i\geq 1$, let $a_i\in S$ and $b_i\in Q$ be the $i$-nearest neighbours of $a\in S$ and $b=g(a)\in Q$, respectively.
Then the Euclidean distances from $a,b$ to their $i$-th neighbours $a_i,b_i$ are $2\ep$-close, i.e. $||a-a_i|-|b-b_i||\leq 2\ep$. 
\bs
\end{lem}
\begin{proof}
Translate the full set $Q$ by the vector $a-g(a)$.
So we assume that $a=g(a)$ and $|b-g(b)|<2\ep$ for all $b\in S$.
Assume by contradiction that the distance from $a$ to its $i$-th neighbour $b_i$ is less than $|a-a_i|-2\ep$.
Then all first $i$ neighbours $b_1,\dots,b_i$ of $a$ within $Q$ belong to the open ball with the center $a$ and the radius $|a-a_i|-2\ep$. 
Since the bijection $g$ shifted every point $b_1,\dots,b_i$ by at most $2\ep$, their preimages $g^{-1}(b_1),\dots,g^{-1}(b_i)$ belong to the open ball with the center $a=g(a)$ and the radius $|a-a_i|$.
Then the $i$-th neighbour of $a$ within $S$ is among these $i$ preimages.
Hence the distance from $a$ to its $i$-th nearest neighbour is strictly less than the required distance $|a-a_i|$.
A similar contradiction is obtained from the assumption that the distance from $a$ to its new $i$-th neighbour $b_i$ is more than $|a-a_i|+2\ep$.
\end{proof}

\begin{theorem}[continuity of $\AMD$ under perturbations]
\label{thm:continuity}
Let finite or periodic sets $S,Q\subset\R^n$ satisfy $d_B(S,Q)<r(Q)$.
Then $|\AMD_k(S)-\AMD_k(Q)|\leq 2d_B(S,Q)$ for any $k\geq 1$.
\bs
\end{theorem}
\begin{proof}%[Proof of Theorem~\ref{thm:continuity}]
By Lemma~\ref{lem:common_lattice} the sets $S,Q$ have a common lattice $\La$.
Any primitive cell $U$ of $\La$ is a unit cell of $S,Q$, i.e. $S=\La+(S\cap U)$ and $Q=\La+(Q\cap U)$.
Since the bottleneck distance $\ep=d_B(S,Q)<r(Q)$, we can define a bijection $g$ from every point $a\in S$ to its unique $\ep$-closest neighbour $g(a)\in Q$. 
%\medskip
If $U$ is a non-primitive unit cell of $S$, the matrix $D(S;k)$ can be constructed as in Definition~\ref{dfn:AMD}, but each row will be repeated $n(S)>1$ times, where $n(S)$ is $\vol[U]$ divided by the volume of a primitive unit cell of $S$.
The average $\AMD_k(S)$ in the $k$-th column is independent of the factor $n(S)>1$.
%\medskip
Since the above conclusions hold for $Q$ instead of $S$, we now compare the matrices $D(S;k)$ and $D(Q;k)$ that are built on the same unit cell $U$ and have equal sizes. 
%\medskip
By Lemma~\ref{lem:perturbed_distances} the corresponding elements of the matrices $D(S;k)$ and $D(Q;k)$ differ by at most $2\ep$, i.e. $|D_{ij}(S;k)-D_{ij}(Q;k)|\leq 2\ep$.
Then the average of the $k$-th column changes by at most $2\ep$, i.e. $|\AMD_k(S)-\AMD_k(Q)|\leq 2\ep$.
\end{proof}

%Theorem~\ref{thm:continuity} and later results are proved in~\ref{sec:proofs}.
 
%5============
\section{The asympotic behaviour of Average Minimum Distances}
\label{sec:asymptotic}

The main result of this section is Theorem~\ref{thm:asymptotic} explicitly describing the asymptotic growth of $\AMD_k$ as $k\to+\infty$ for a wide class of sets including non-periodic sets.
$\AMD_k(S)$ approaches $c(S)\sqrt[n]{k}$.
The point packing coefficient $c(S)$ is introduced below, see examples in Fig.~\ref{fig:lattices2D}.
\medskip

The volume of the unit ball in $\R^n$ is $V_n=\dfrac{\pi^{n/2}}{\Ga(\frac{n}{2}+1)}$, where $\Ga$ denotes Euler's Gamma function $\Ga(m)=(m-1)!$ and $\Ga(\frac{m}{2}+1)=\sqrt{\pi}(m-\frac{1}{2})(m-\frac{3}{2})\cdots\frac{1}{2}$ for any integer $m\geq 1$.

\begin{dfn}[$(U,m)$-sets $S$, $\AMD_k(S;U)$, the point packing coefficient $c(S)$]
\label{dfn:Um-set}
Let $U$ be a unit cell of a lattice $\La\subset\R^n$.
For any fixed $m\geq 1$, a set $S\subset\R^n$ is called a \emph{$(U,m)$-set} if $S\cap (U+\vec v)$ consists of $m$ points for any vector $\vec v\in\La$.
%\medskip
For any point $p\in S\cap U$, let $d_k(S;p)$ be the distance from $p$ to its $k$-th nearest neighbour in $S$.
% as in Definition~\ref{dfn:AMD}.
The \emph{Average Minimum Distance} is $\AMD_k(S;U)=\dfrac{1}{m}\sum\limits_{p\in S\cap U} d_k(S;p)$.
The \emph{Point Packing Coefficient} is $c(S)=\sqrt[n]{\dfrac{\vol[U]}{mV_n}}$.
\bs
\end{dfn}

\newcommand{\size}{17mm}
\begin{figure}[h!]
\centering
\includegraphics[height=\size]{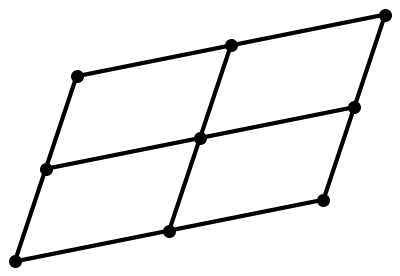}
\hspace*{0mm}
\includegraphics[height=\size]{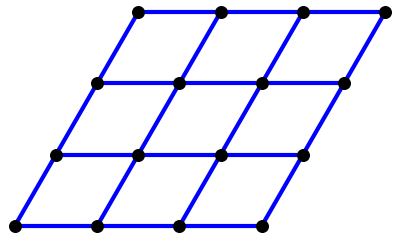}
\hspace*{0mm}
\includegraphics[height=\size]{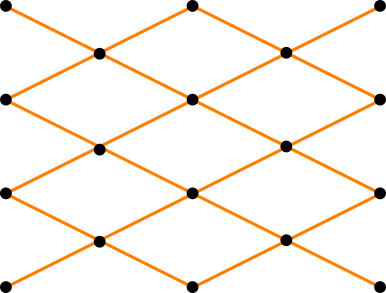}
\includegraphics[height=\size]{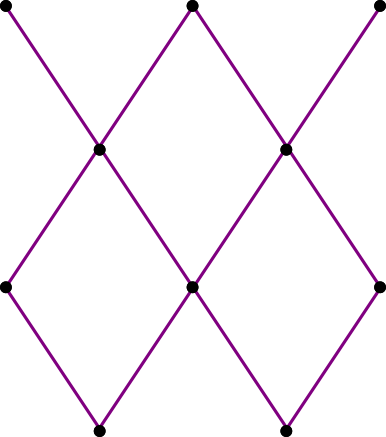}
\hspace*{0mm}
%% TO BE CHANGED
\includegraphics[height=\size]{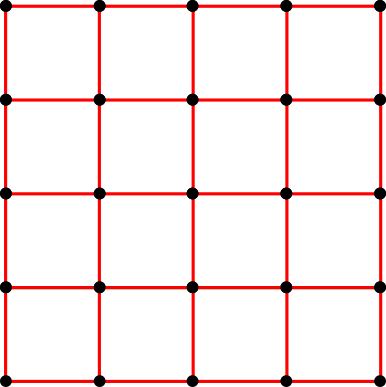}
\hspace*{0mm}
\includegraphics[height=\size]{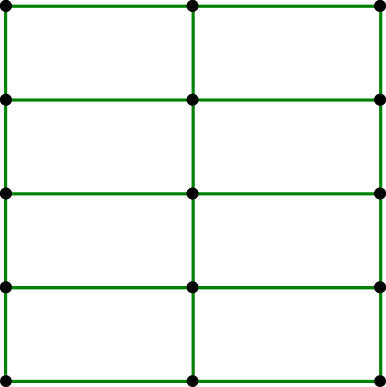}
\caption{
The six 2-dimensional lattices $\La_i$ have point packing coefficients $c(\La_i)$ below and curves of $\AMD_k$ in Fig.~\ref{fig:lattices2D_AMD30+300}. 
\textbf{1st}: a generic black lattice $\La_1$ with the basis $(1.25,0.25),(0.25,0.75)$ and $c(\La_1)=\sqrt{\dfrac{7}{8\pi}}\approx 0.525$.
\textbf{2nd}: the blue hexagonal lattice $\La_2$ with the basis $(1,0),(1/2,\sqrt{3}/2)$ and $c(\La_2)=\sqrt{\dfrac{\sqrt{3}}{2\pi}}\approx 0.528$.
\textbf{3rd}: the orange rhombic lattice $\La_3$ with the basis $(1,0.5),(1,-0.5)$ and $c(\La_3)=\sqrt{\dfrac{1}{\pi}}\approx 0.564$.
\textbf{4th}: the purple rhombic lattice $\La_4$ with the basis $(1,1.5),(1,-1.5)$ and $c(\La_4)=\sqrt{\dfrac{3}{\pi}}\approx 0.977$.
\textbf{5th}: the red square lattice $\La_5$ with the basis $(1,0),(0,1)$ and $c(\La_5)=\sqrt{\dfrac{1}{\pi}}\approx 0.564$. 
\textbf{6th}: the green rectangular lattice $\La_6$ with the basis $(2,0),(0,1)$ and $c(\La_6)=\sqrt{\dfrac{2}{\pi}}\approx 0.798$.
}
\label{fig:lattices2D}
\end{figure}

\begin{figure}[h!]
\centering
\includegraphics[width=0.49\linewidth]{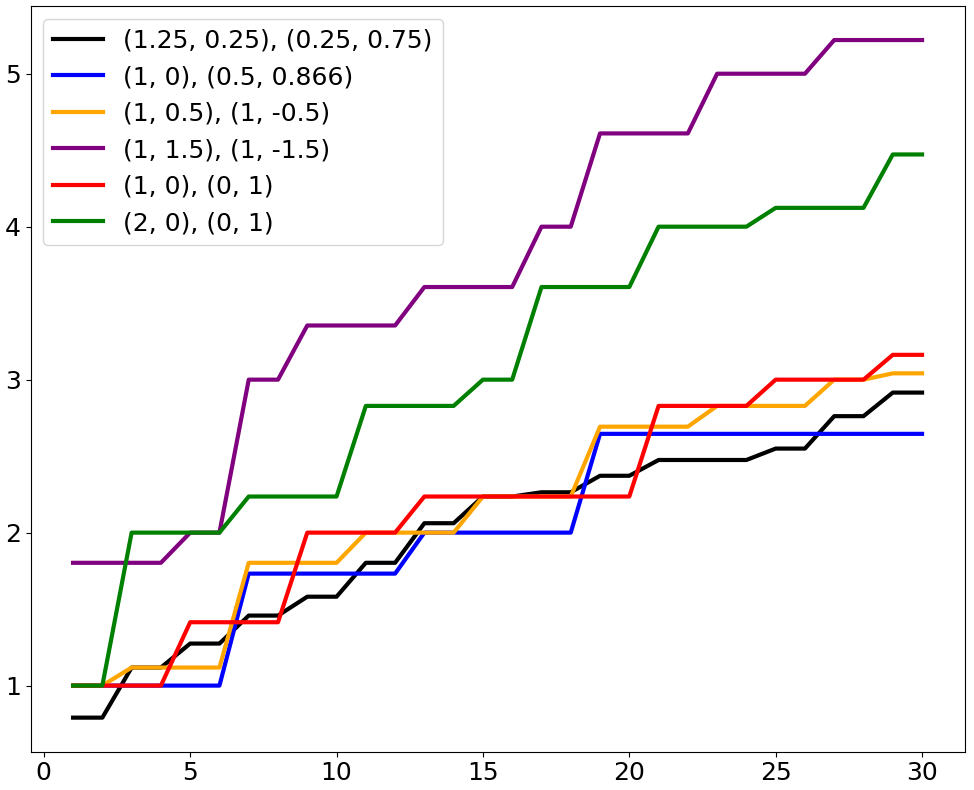}
\includegraphics[width=0.49\linewidth]{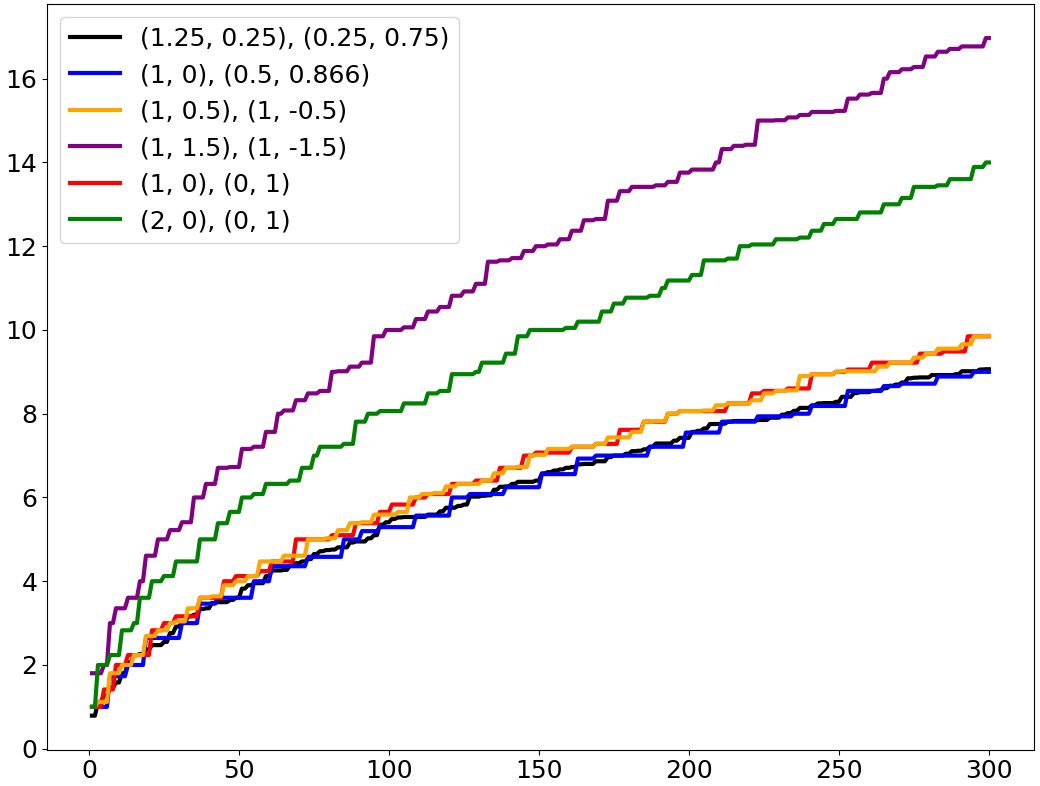}
\caption{\textbf{Left}: $\AMD_k$ up to $k=30$ for the 2D lattices in Fig.~\ref{fig:lattices2D}. 
\textbf{Right}: $\AMD_k$ extended to $k=300$.
Though the orange and red lattices have close point packing coefficients $c(\La_i)$, so their AMD curves approach each other by Theorem~\ref{thm:asymptotic}, $\AMD_k$ for small $k$ distinguish these lattices.}
\label{fig:lattices2D_AMD30+300}
\end{figure}

A broad example of a $(U,m)$-set is a non-periodic perturbation of a periodic set $S=\La+M$, where a lattice $\La$ is generated by a unit cell $U$, a motif $M$ has $m$ points. 
Since $S$ may not be periodic, $\AMD_k(S;U+\vec v)$ can depend on a shift vector $\vec v\in\La$.
Even if $S$ is periodic, a unit cell $U$ in Definition~\ref{dfn:Um-set} can be non-primitive. 
However, $\vol[U]/m$ is independent of a choice of $U$.
Hence $c(S)$ is an isometry invariant, also for any $(U,m)$-set $S$, because a shifted cell $U+\vec v$ contains the same number $m$ of points from $S$.
If all points have the weight $V_n$ of the unit ball in $\R^n$, then $(c(S))^n$ is inversely proportional to the density $\rho$ of $S$.
Lemmas~\ref{lem:ball_size} and~\ref{lem:distance_bounds} are needed to prove Theorem~\ref{thm:asymptotic}.
The \emph{diameter} of a unit cell $U$ is $d=\sup\limits_{a,b\in U}|a-b|$.

\begin{lem}[bounds on points within a ball]
\label{lem:ball_size}
Let $S\subset\R^n$ be any $(U,m)$-set with a unit cell $U$, which generates a lattice $\La$ and has a diameter $d$.
For any point $p\in S\cap U$ and a radius $r$, 
%the unions ${U'}(p;r)$ and ${U''}(p;r)$ of shifted cells that bound $\bar B(p;r)$ as follows:
% replaced
% $$\text{consider the \emph{lower} union }U'(p;r)=\bigcup \{(U+\vec v) \text{ such that } \vec v\in\La, (U+\vec v)\subset\bar B(p;r)\}$$
% $$\text{and the \emph{upper} union }U''(p;r)=\bigcup \{(U+\vec v) \text{ such that }  \vec v\in\La, (U+\vec v)\cap\bar  B(p;r)\neq\emptyset\}.$$
consider the \emph{lower} union  
$U'(p;r)=\bigcup \{(U+\vec v) \text{ such that } \vec v\in\La, (U+\vec v)\subset\bar B(p;r)\}$
and the \emph{upper} union 
$U''(p;r)=\bigcup \{(U+\vec v) \text{ such that }  \vec v\in\La, (U+\vec v)\cap\bar  B(p;r)\neq\emptyset\}$.
Then the number of points from $S$ in the closed ball $\bar B(p;r)$ with center $p$ and radius $r$ has the bounds
$\left(\dfrac{r-d}{c(S)}\right)^n\leq
m\dfrac{\vol[U'(p;r)]}{\vol[U]}\leq|S\cap\bar B(p;r)|\leq 
m\dfrac{\vol[U''(p;r)]}{\vol[U]}\leq
\left(\dfrac{r+d}{c(S)}\right)^n$.
%, where $c(S)$ is the Point Packing Coefficient introduced for any $(U,m)$-set $S$ in Definition~\ref{dfn:Um-set}.
\bs
\end{lem}
\begin{proof}
Intersect the three regions ${U'}(p;r)\subset\bar B(p;r)\subset{U''}(p;r)$ with $S$ in $\R^n$ and count the numbers of resulting points:
$|S\cap{U'}(p;r)|\leq |S\cap\bar B(p;r)|\leq |S\cap{U''}(p;r)|$.
\medskip

The union ${U'}(p;r)$ consists of $\dfrac{\vol[U'(p;r)]}{\vol[U]}$ cells, which all have the same volume $\vol[U]$.
Since $|S\cap U|=m$, we now get $|S\cap U'(p;r)|=m\dfrac{\vol[U'(p;r)]}{\vol[U]}$.
Similarly we count the points in the upper union:
$|S\cap{U''}(p;r)|=m\dfrac{\vol[{U''}(p;r)]}{\vol[U]}$.
The bounds of $|S\cap\bar  B(p;r)|$ become
$$m\dfrac{\vol[{U'}(p;r)]}{\vol[U]}\leq|S\cap\bar B(p;r)|\leq
m\dfrac{\vol[{U''}(p;r)]}{\vol[U]},$$
$$\vol[{U'}(p;r)]\leq\dfrac{\vol[U]}{m}|S\cap\bar B(p;r)|\leq\vol[{U''}(p;r)].$$
For the diameter $d$ of the unit cell $U$, the smaller ball $\bar B(p;r-d)$ is completely contained within the lower union ${U'}(p;r)$.
Indeed, if $|\vec q-\vec p|\leq r-d$, then $q\in U+\vec v$ for some $\vec v\in\La$.
Then $(U+\vec v)$ is covered by the ball $\bar B(q;d)$, hence by $\bar B(p;r)$ due to the triangle inequality.
\medskip

The inclusion $\bar B(p;r-d)\subset{U'}(p;r)$ implies the lower bound for the volumes: 
$$V_n(r-d)^n=\vol[\bar B(p;r-d)]\leq\vol[{U'}(p;r)]\text{  where},$$ $V_n$ is the volume of the unit ball in $\R^n$.
The similar inclusion $U''(p;r)\subset\bar B(p;r+d)$ gives 
$$\vol[{U''}(p;r)]\leq\vol[\bar B(p;r+d)]=V_n(r+d)^n,\;
V_n(r-d)^n\leq\dfrac{\vol[U]}{m}|S\cap B(p;r)|\leq V_n(r+d)^n,$$
$\dfrac{mV_n}{\vol[U]}(r-d)^n\leq|S\cap B(p;r)|\leq\dfrac{mV_n}{\vol[U]}(r+d)^n$,
which implies the required inequality.
% with $c(S)=\sqrt[n]{\dfrac{\vol[U]}{mV_n}}$.
\end{proof}

\begin{lem}[distance bounds]
\label{lem:distance_bounds}
Let $S\subset\R^n$ be any $(U,m)$-set with a unit cell $U$ of diameter $d$, see Definition~\ref{dfn:Um-set}.
For any point $p\in S\cap U$, let $d_k(S;p)$ be the distance from $p$ to its $k$-th nearest neighbour in $S$.
Then %$|d_k(S,p)-c(S)\sqrt[n]{k}|<d$
$c(S)\sqrt[n]{k}-d<d_k(S;p)\leq c(S)\sqrt[n]{k}+d$ for any $k\geq 1$.
\bs
\end{lem}
\begin{proof}
The closed ball $\bar B(p;r)$ of radius $r=d_k(S;p)$ has more than $k$ points (including $p$) from $S$.
The upper bound of Lemma~\ref{lem:ball_size} for $r=d_k(S;p)$ implies that $k<|S\cap\bar  B(p;r)|\leq\dfrac{(r+d)^n}{(c(S))^n}$.
Taking the $n$-th root of both sides, we get
$\sqrt[n]{k}<\dfrac{r+d}{c(S)}$, so
$r=d_k(S;p)>c(S)\sqrt[n]{k}-d$.
\medskip

For any smaller radius $r<d_k(S;p)$, the closed ball $\bar B(p;r)$ contains at most $k$ points (including $p$) from $S$.
The lower bound of Lemma~\ref{lem:ball_size} for any $r<d_k(S;p)$ implies that $\dfrac{(r-d)^n}{c(S)^n}\leq |S\cap\bar B(p;r)|\leq k$.
Since $\dfrac{(r-d)^n}{c(S)^n}\leq k$ holds for the constant upper bound $k$ and any radius $r<d_k(S;p)$, the same inequality holds for  the radius $r=d_k(S;p)$.
\medskip

Similarly to the upper bound above, we get  
$\dfrac{r-d}{c(S)}\leq\sqrt[n]{k}$, 
$r=d_k(S;p)\leq c(S)\sqrt[n]{k}+d$.
Combine the two bounds above as follows: $c(S)\sqrt[n]{k}-d<d_k(S;p)\leq c(S)\sqrt[n]{k}+d$.
\end{proof}

\begin{theorem}[asymptotic behaviour of AMD]
\label{thm:asymptotic}
For any $(U,m)$-set $S\subset\R^n$ from Definition~\ref{dfn:Um-set}, we have $|\AMD_k(S;U)-c(S)\sqrt[n]{k}|\leq d$ for any $k\geq 1$ and
$\lim\limits_{k\to+\infty}\dfrac{\AMD_k(S;U)}{\sqrt[n]{k}}=c(S)$.
\bs
\end{theorem}
\begin{proof}%[Proof of Theorem~\ref{thm:asymptotic}]
Averaging the bounds of Lemma~\ref{lem:distance_bounds} over all points $p\in S\cap U$, we get the required bounds:
$c(S)\sqrt[n]{k}-d<\AMD_k(S;U)=\dfrac{1}{m}\sum\limits_{p\in S\cap U}d_k(S;p)\leq c(S)\sqrt[n]{k}+d$, which imply that
$|\AMD_k(S;U)-c(S)\sqrt[n]{k}|\leq d$ for $k\geq 1$.
Hence
$\lim\limits_{k\to+\infty}\dfrac{\AMD_k(S;U)}{\sqrt[n]{k}}=c(S)$.
\end{proof}

%6============
\section{A near linear time algorithm, experiments and a discussion}
\label{sec:discussion}

This section describes the AMD algorithm in Theorem~\ref{thm:algorithm} and experiments on big datasets.
\medskip

%then discusses its application to real and simulated crystals.
% The input for computing the an AMD is a periodic point set $S$ given by (parameters of) its unit cell $U$ and coordinates of $m$ motif points in the basis of $U$.
The input for computing the an AMD is a periodic point set $S$ given by the basis vectors of its unit cell $U$ and the Cartesian coordinates of $m$ motif points.
%This information is contained in a Crystallographic Information File (CIF).
The length $k$ of the vector $\AMD^{(k)}(S)=(\AMD_1,\dots,\AMD_k)$ is independent of a periodic point set $S$.
Increasing $k$ adds more components to the vector $\AMD^{(k)}$ without changing any previous values.
\medskip

The size of an input for any real periodic structure is proportional to the number $m$ of points in a motif.
Theorem~\ref{thm:algorithm} solves Problem~\ref{pro:invariants} requiring a near linear time in both $k,m$.

\begin{theorem}[a near linear time algorithm for AMD]
\label{thm:algorithm}
Let a periodic set $S\subset\R^n$ have $m$ points in a unit cell $U$.
For a fixed dimension $n$ and $i=1,\dots,k$, $\AMD_i(S)$ can be computed in a time $O(\nu(S;n)km\log(km))$, where $\nu(S;n)$ is independent of $k,m$, see Fig.~\ref{fig:time}.
\bs
\end{theorem}
\begin{proof}%[Proof of Theorem~\ref{thm:algorithm}]
Let the origin $0$ be in the center of the unit cell $U$, which generates a lattice $\La$.
If $d$ is the diameter of $U$, any point $p\in M=S\cap U$ is covered by the closed ball $\bar B(0,0.5d)$.
\medskip

By Lemma~\ref{lem:distance_bounds} all $k$ neighbours of $p$ are covered by the ball $\bar B(0;r)$ of radius $r=c(S)\sqrt[n]{k}+1.5d$.
To generate all $\La$-translates of $M$ within $\bar B(0;r)$, we gradually extend $U$ in spherical layers by adding more shifted cells until we get the upper union $U''(0;r)\supset\bar B(0;r)$.
By Lemma~\ref{lem:ball_size} the union $U''(0;r)$ 
includes all $k$ neighbours of motif points and has at most
$$\mu
=m\dfrac{\vol[U''(0;r)]}{\vol[U]}
\leq\left(\dfrac{c(S)\sqrt[n]{k}+2.5d}{c(S)}\right)^n
=\left(\sqrt[n]{k}+\dfrac{2.5d}{c(S)}\right)^n=O(f(S;n)k) \text{ points,}$$
%\begin{align*}\mu =m\dfrac{\vol[U''(0;r)]}{\vol[U]} \leq\left(\dfrac{c(S)\sqrt[n]{k}+2.5d}{c(S)}\right)^n =\left(\sqrt[n]{k}+\dfrac{2.5d}{c(S)}\right)^n=O(f(S;n)k)\end{align*}
where $f(S;n)$ denotes a suitable function that is independent of $k\to+\infty$ for a fixed dimension $n$. 
%\medskip
We build an $n$-d tree \cite{Brown2015kdtree} on $\mu$ points above in time $O(n\mu\log\mu)$.
For each $p\in M$, all $k$ neighbours of $p$ can be found and ordered by distances to $p$ in time $O(\mu\log\mu)$.
%\medskip

By Definition~\ref{dfn:AMD} we lexicographically sort $m$ lists of ordered distances in time $O(km\log m)$, because a comparison of any two ordered lists of length $k$ takes $O(k)$ time.
%\medskip

The ordered lists of distances are the rows of the matrix $D(S;k)$.
All $\AMD_i(S)$ are found in time $O(km)$.
The total time is $O((m+n)\mu\log\mu+km\log m)=O(\nu(S;n)km\log(km))$.
%, where $\nu$ is independent of $k,m$.
\end{proof}

\begin{figure}[h!]
\includegraphics[width=0.49\linewidth]{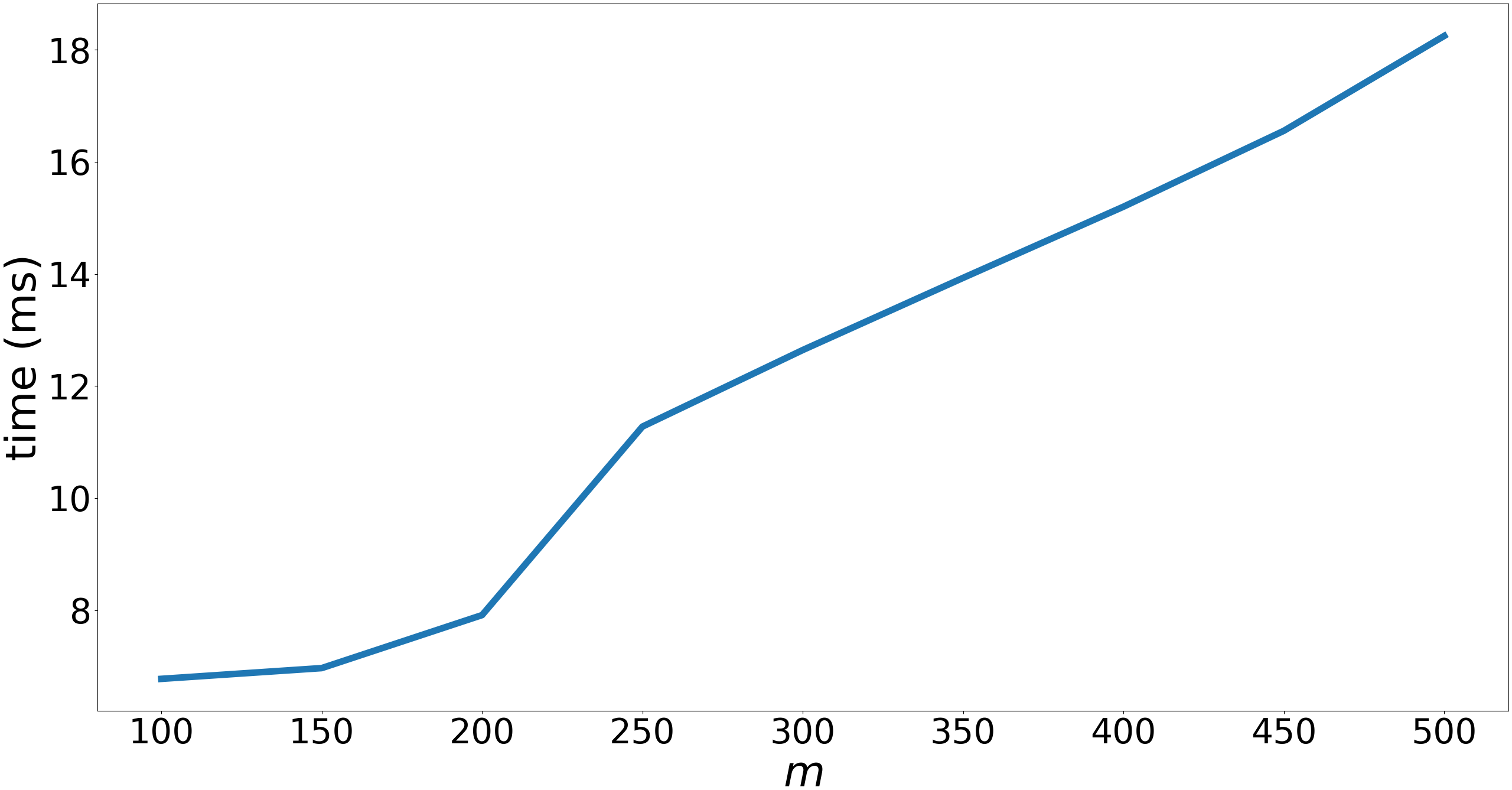}
\includegraphics[width=0.49\linewidth]{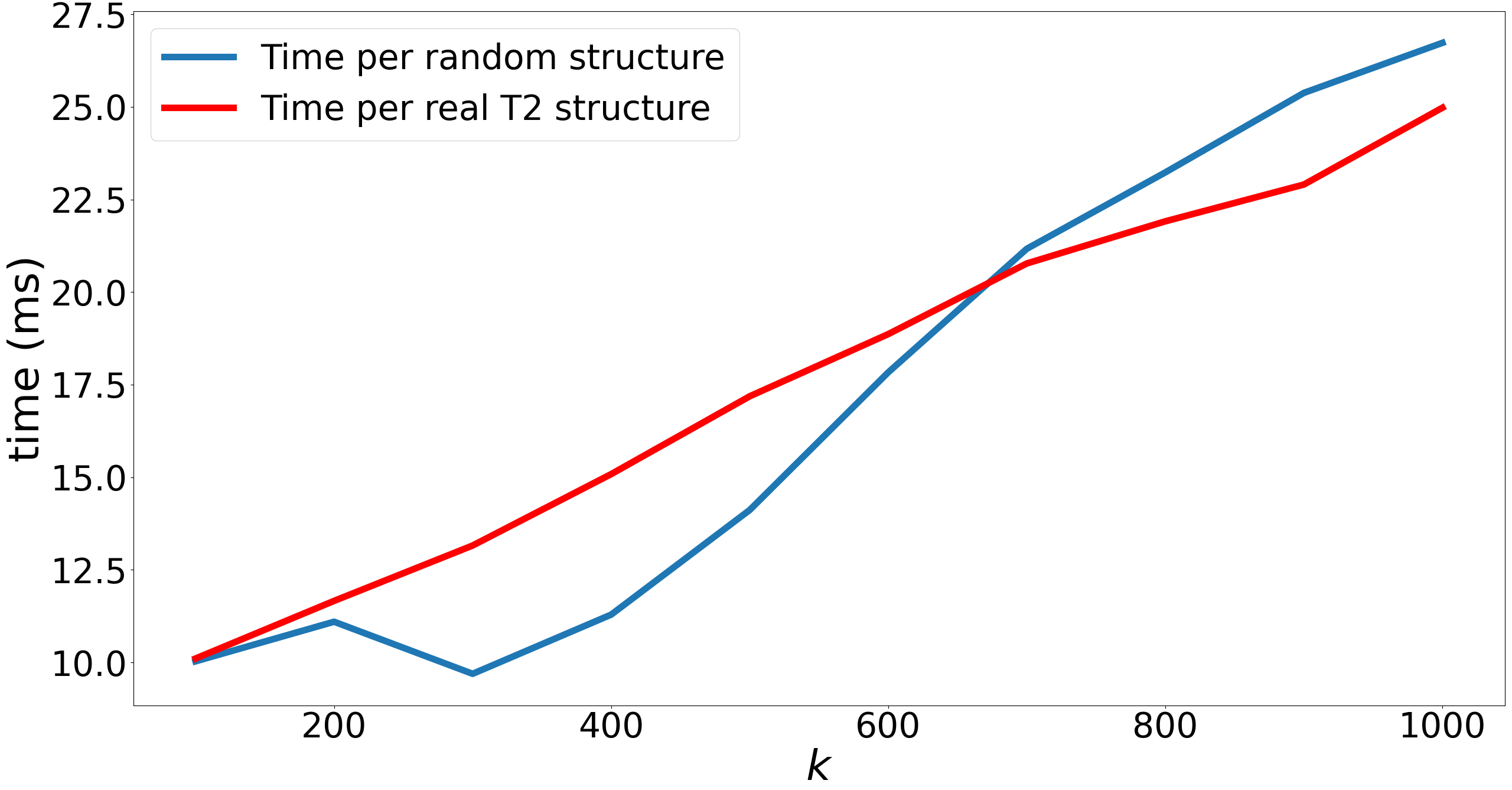}
\caption{
The blue curves show the AMD time in milliseconds averaged over 3000 sets whose points are uniformly generated in unit cells with random edge-lengths in $[1,2]$ and angles in $[\frac{\pi}{3},\frac{2\pi}{3}]$.
\textbf{Left}: $\AMD^{(200)}$.
\textbf{Right}: the blue curve is for $m=250$, the red curve is for 5679 real structures.}
\label{fig:time}
\end{figure}

Fig.~\ref{fig:time} illustrates a near linear running time justified by Theorem~\ref{thm:algorithm}. 
% The graph on the left was obtained on 3000 random periodic sets with uniformly chosen parameters of a unit cell and coordinates of points in a motif.
The red curve on the right was obtained on 5679 predicted  structures, which have the same chemical compositions and contain about 250 points on average in a unit cell, as reported in our past paper \cite{pulido2017functional}.
\medskip

The only other implemented infinite sequence of continuous isometry invariants are density functions $\psi_k(t)$ requiring a cubic time in $k$, see \cite[section~6.3]{edels2021}. 
Because of this cubic increase, we ran the C++ code only up to $k=8$ for four representative points out of 46 atoms per molecule.
These simplified sets contain about 25 points on average in a unit cell.  
\medskip

Each $\psi_k(t)$ depends on a continuous radius $t$, and hence can be evaluated only at discretely sampled values of $t$.
We computed the first eight $\psi_k(t)$ sampled at $t$ equal to multiples of 0.2$\AA$ up to about 12$\AA$, where 1$\AA=10^{-10}$m.
Smallest distances between points in real structures are about 0.8$\AA$ at this atomic scale. 
With the above parameters, the experiments on similar machines
(Dell XPS 15 6-core, 2.20GHz, 16GB and MacBook Pro, 2.3GHz, 8GB) took about 1 min per structure, more than 4 days for 5679 simplified sets.
\medskip

In comparison with the above times of density functions, vectors $\AMD^{(k)}$ were computed for full periodic structures and much larger $k$ on the modest desktop AMD Ryzen 5 6-core 4.60Ghz, 32GB DDR4.
The red curve in Fig.~\ref{fig:time} for the 5679 structures implies the average time of 10ms for $\AMD^{(100)}$ and 27ms for $\AMD^{(1000)}$, so the total time is about 17 min.
Despite the Cambridge Structural Database (CSD) having much more diverse compositions in comparison with T2 structures based on the same molecule, $\AMD^{(100)}$ required the similar time of 13.4ms on average, less than 52 min in total for all 228,994 organic structures.
\medskip

We computed the $L_{\infty}$-distance between $\AMD^{(100)}$ vectors of all structures. 
A Python script making use of vectorization in NumPy was able to compute all $\binom{228972}{2}$ distances in 8 hours 52 min, about 1.2 microseconds on average.
A final visualisation of nearly 229K structures as a Minimum Spanning Tree in Fig.~\ref{fig:CSDorganic_AMD100TMap} in Appendix~D took less than 1 min.
\medskip

%7============
%\section{A discussion of novel contributions and further developments}
%\label{sec:discussion}

In conclusion, the AMD sequence is a novel continuous invariant of periodic point sets up to isometry in $\R^n$ by Theorems~\ref{thm:AMD_invariant} and~\ref{thm:continuity}.
Theorem~\ref{thm:asymptotic} is especially strong due to the asymptotic formula working even for non-periodic sets $S$ for the newly defined point packing coefficient $c(S)$.
Theorem~\ref{thm:algorithm} justified a near linear running time in both input parameters $k,m$, which has enabled visualisations of huge and diverse datasets with minimal resources.
\medskip

The key motivation for continuous isometry invariants is the discontinuity of reduced cells, which was experimentally known \cite{andrews1980perturbation} since 1980.
Let $L$ be the space of all lattices in $\R^n$. 
Let $B$ be the space of linear bases $b=\{\vec v_1,\dots,\vec v_n\}$.
If we concatenate $n$ vectors into one vector with $n^2$ coordinates, the space $B$ becomes a subset of $\R^{n^2}$ with the Euclidean topology.  
Since any basis generates a lattice by Definition~\ref{dfn:crystal}, we have the projection $g:B\to L$.
%For any fixed lattice $\La\in L$, the preimage $g^{-1}(\La)$ consists of all primitive bases of $\La$.
Consider a minimal topology on $L$ that makes $g$ continuous so that the preimage $g^{-1}(N(\La))$ of any open neighbourhood $N(\La)$ is a union of open neighbourhoods of bases from $g^{-1}(\La)\subset B$.
%Then continuity of $g:B\to L$ means that any small perturbation of a basis gives rise to a small perturbation of the lattice generated by this basis.
A desired reduction would be a continuous map $h:L\to B$ such that $g\circ h(\La)=\La$ is the identity map.
In other words, if we continuously change a lattice $\La\in L$, its reduced basis $h(\La)\in B$ should also change continuously.
Theorem~\ref{thm:discontinuity} disproves any possibility of a continuous reduction.

\begin{theorem}[discontinuity of reduced cells]
\label{thm:discontinuity}
Let $g:B\to L$ map any basis $b$ of $\R^n$ to its lattice $\La$.
Then there is no continuous map $h:L\to B$ such that $g\circ h$ is the identity.
\bs
\end{theorem}

A limitation of $\AMD$ is its potential incompleteness, though we do not know non-isometric periodic sets that have equal $\AMD_k$ for all $k$.
Theorem~\ref{thm:completeness} (proved with Theorem~\ref{thm:discontinuity} in appendix~A) hints at completeness of $\AMD$ for periodic point sets in general position. 

\begin{theorem}[completeness for generic finite sets]
\label{thm:completeness}
Let a finite set $S\subset\R^n$ consist of $m$ points such that all pairwise distances between points of $S$ are distinct.
Then $S$ can be uniquely reconstructed  from the matrix $D(S;m-1)$ in Definition~\ref{dfn:AMD} up to isometry of $\R^n$.
\bs
\end{theorem}

The code for AMDs in \cite{AMD_DW,AMD_MM} produce identical results. 
We thank the co-authors of \cite{edels2021} and Janos Pach for fruitful discussions and all reviewers in advance for helpful suggestions.\newpage

\bibliography{MFCS2021AMD}

%\end{document}
\newpage

\appendix
%A=============
\section{Extra examples and proofs of Lemma~\ref{lem:common_lattice} and Theorems~\ref{thm:discontinuity},~\ref{thm:completeness}}
\label{sec:proofs}

Following the main 12-page paper and appendix A with extra examples and proofs for mathematicians and computer scientists, all remaining appendices B,C,D are especially written for experts in applied areas such as crystallography and materials science.
The breakthrough visualisations in Fig.~\ref{fig:T2AMD200TMap}, \ref{fig:CSD_drugs_TMap_AMD200},\ref{fig:CSDorganic_AMD100TMap} became possible due to the fast AMD.
\medskip

%This appendix shows the strength of AMD invariants  in Examples~\ref{exa:SQ15},~\ref{exa:SQ32} and then rigorously proves Theorems~\ref{thm:continuity}, \ref{thm:asymptotic}, \ref{thm:algorithm}.
The density functions narrowly distinguish the following periodic sets \cite[Example~1]{edels2021}
$$S_{32} = \{ 0, 7, 8, 9, 12, 15, 17, 18, 19, 20, 21, 22, 26, 27, 29, 30 \} + 32\Z,$$
$$Q_{32} = \{ 0, 1, 8, 9, 10, 12, 13, 15, 18, 19, 20, 21, 22, 23, 27, 30 \} + 32\Z$$
 with period 32 in Fig.~\ref{fig:SQ32}.
Tables~\ref{tab:S32},~\ref{tab:Q32} distinguish these very similar sets by $\AMD_k$ for $k=2,3$.

\begin{figure}[h!]
\includegraphics[width=0.8\linewidth]{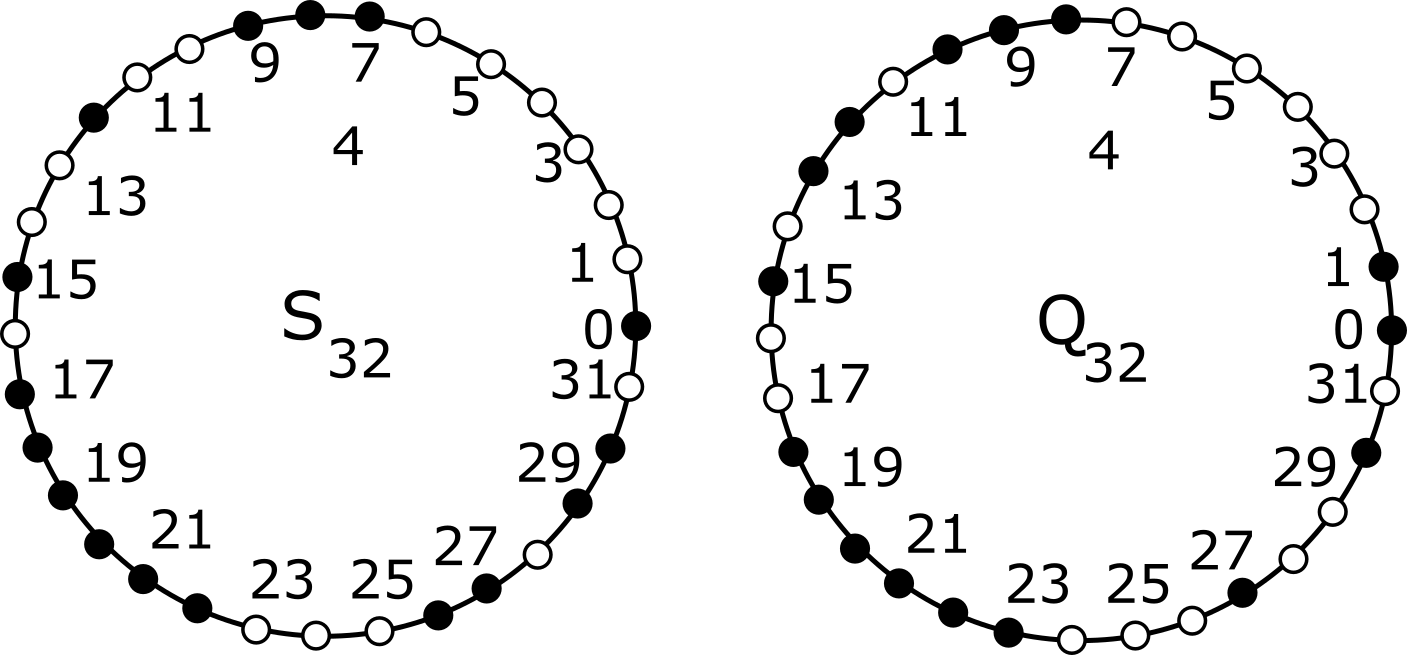}
\caption{Circular versions of the periodic point sets $S_{32},Q_{32}$.
% in Example~\ref{exa:SQ32}.
Distances between any points are measured along round arcs in each circle.}
\label{fig:SQ32}      
\end{figure}

\begin{table*}[h!]
\begin{tabular}{C{10mm}|
C{2mm}C{2mm}C{2mm}C{2mm}C{2mm}C{2mm}C{2mm}C{2mm}
C{2mm}C{2mm}C{2mm}C{2mm}C{2mm}C{2mm}C{2mm}C{2mm}|C{10mm}}
\hline\noalign{\smallskip}
$S_{32}$ & 0 & 7 & 8 & 9 & 12 & 15 & 17 & 18 & 19 & 20 & 21 & 22 & 26 & 27 & 29 & 30 & $\AMD_k$ \\
\hline
\noalign{\smallskip}\noalign{\smallskip}
$k=1$     & 2 & 1 & 1 & 1 &  3  &  2  &  1  &  1 &   1  &  1  &  1  &  1 &  1  &   1 &   1 &  1 & 20/16 \\
\noalign{\smallskip}\noalign{\smallskip}
$k=2$     & 3 & 2 & 1 & 2 &  3  &  3  &  2  &  1 &   1  &  1  &  1  &  2 &  3  &   2 &   2 &  2 & 31/16 \\
\noalign{\smallskip}\noalign{\smallskip}
$k=3$     & 5 & 5 & 4 & 3 &  4  &  3  &  2  &  2 &   2  &  2  &  2  &  3 &  4  &   3 &   3 &  3 & 50/16
\end{tabular}
\caption{\textbf{First row}: 16 points from the motif $M\subset S_{32}$ in Fig.~\ref{fig:SQ32}. 
\textbf{Further rows}: distance from each point $p\in M$ to its $k$-th nearest neighbour in $S_{32}$.}
\label{tab:S32}
\end{table*}

\begin{table*}[h!]
\begin{tabular}{C{10mm}|
C{2mm}C{2mm}C{2mm}C{2mm}C{2mm}C{2mm}C{2mm}C{2mm}
C{2mm}C{2mm}C{2mm}C{2mm}C{2mm}C{2mm}C{2mm}C{2mm}||C{10mm}}
\hline\noalign{\smallskip}
$Q_{32}$ & 0 & 1 & 8 & 9 & 10 & 12 & 13 & 15 & 18 & 19 & 20 & 21 & 22 & 23 & 27 & 30 & $\AMD_k$ \\
\hline
\noalign{\smallskip}\noalign{\smallskip}
$k=1$      & 1 & 1 & 1 & 1 &  1  &  1  &  1  &  2 &   1  &  1  &  1  &  1 &  1  &   1 &   3 &  2 & 20/16 \\
\noalign{\smallskip}\noalign{\smallskip}
$k=2$      & 2 & 3 & 2 & 1 &  2  &  2  &  2  &  3 &   2  &  1  &  1  &  1 &  1  &   2 &   4 &  3 & 32/16 \\
\noalign{\smallskip}\noalign{\smallskip}
$k=3$      & 5 & 6 & 4 & 3 &  2  &  3  &  3  &  3 &   3  &  2  &  2  &  2 &  2  &   3 &   5 &  3 & 51/16 \\
\end{tabular}
\caption{\textbf{First row}: 16 points from the motif $M\subset Q_{32}$. 
\textbf{Further rows}: distance from each point $p\in M$ to its $k$-th nearest neighbour in $Q_{32}$, see Fig.~\ref{fig:SQ32}.}
\label{tab:Q32}
\end{table*}

\begin{proof}[Proof of Lemma~\ref{lem:common_lattice}]
Let $S=\La(S)+(U(S)\cap S)$ and $Q=\La(Q)+(U(Q)\cap Q)$, where 
$U(S),U(Q)$ are initial unit cells of $S,Q$ and the lattices $\La(S),\La(Q)$ contain the origin.
\medskip

By shifting all points of $S,Q$ (but not their lattices), we guarantee that $S$ contains the origin $0$ of $\R^n$.
Assume by contradiction that the given periodic point sets $S,Q$ have no common lattice.
Then there is a vector $\vec p\in\La(S)$ whose all integer multiples $k\vec p\not\in\La(Q)$ for $k\in\Z-0$.
Any such multiple $k\vec p$ can be translated by a vector $\vec v(k)\in\La(Q)$ to the initial unit cell $U(Q)$ so that $\vec q(k)=k\vec p-\vec v(k)\in U(Q)$.
\medskip

Since $U(Q)$ contains infinitely many points $\vec q(k)$,
one can find a pair $\vec q(i),\vec q(j)$ at a distance less than $\de=r(Q)-d_B(S,Q)>0$.
The formula $\vec q(k)\equiv k\vec p\pmod{\La(Q)}$ implies that 
 $\vec q(i+k(j-i)) \equiv (i+k(j-i))\vec p\pmod{\La(Q)} \equiv
 \vec q(i) + k(\vec q(j)-\vec q(i))\pmod{\La(Q)}$.
\medskip
 
If the point $\vec q(i) + k(\vec q(j)-\vec q(i))$ belongs to $U(Q)$, we get the equality $\vec q(i+k(j-i))=\vec q(i) + k(\vec q(j)-\vec q(i))$.
All these points over $k\in\Z$ lie on a straight line within $U(Q)$ and have the distance $|\vec q(j)-\vec q(i)|<\de$ between successive points.
\medskip

The closed balls with radius $d_B(S,Q)$ and centers at points in $Q$ are at least $2\de$ away from each other.
Then one of the points $\vec q(i+k(j-i))$ is more than $d_B(S,Q)$ away from $Q$. 
Hence the point $(i+k(j-i))\vec p\in S$ also has a distance more than $d_B(S,Q)$ from any point of $Q$, which contradicts Definition~\ref{dfn:BND}. % of the bottleneck distance.  
\end{proof}

\begin{proof}[Proof of Theorem~\ref{thm:discontinuity} by contradiction]
Let $h(\Z^n)=\{\vec v_1,\dots,\vec v_n\}\in B$ be a reduced basis of integer lattice $\Z^n\subset\R^n$.
Consider the continuous path $\ga:[0,1]\to B$, where $\ga(t)$ is the basis $\vec v_1+t\vec v_2,\vec v_2,\dots,\vec v_n$.
Since the bases $\ga(0)\neq\ga(1)$ define the same integer lattice $\Z^n\subset\R^n$, the composition $g\circ\ga$ can be considered as a continuous loop $g\circ\ga:[0,1]\to L$ in the space of lattices with $g\circ\ga(0)=\Z^n=g\circ\ga(1)$.
\medskip

It remains to show that the continuous map $h$ lifts the loop $g\circ\ga$ to the path $\ga:[0,1]\to B$ with disjoint endpoints $\ga(0)\neq\ga(1)$, which is a contradiction with the existence of such $h$.
\medskip

Since $h$ is continuous, for all sufficiently small $t>0$, the basis $h\circ g\circ\ga(t)$ should be close to $h(\Z^n)$, hence should coincide with $\ga(t)$, because all other bases of the lattice $h\circ g\circ\ga(t)$ close to $\Z^n$ are sufficiently away from $\ga(t)$ in Euclidean metric on $B$.
This local extension argument works around any $t\in[0,1]$ where we already know that $h\circ g\circ\ga(t)=\ga(t)$.
\medskip

Since any infinite cover of the compact line segment $[0,1]$ by open neighbourhoods contains a finite subcover, we need only finitely many steps to get $\ga(1)=h\circ g\circ\ga(1)=h(\Z^n)=h\circ g\circ\ga(0)=\ga(0)$, which contradicts the fact that the initial bases $\ga(1)\neq\ga(0)$.
\end{proof}

\begin{proof}[Proof of Theorem~\ref{thm:completeness}]
Since the distances between all $m$ points of $S\subset\R^n$ are distinct, every distance appears in the matrix $D(S;m-1)$ exactly twice,
 once as the distance from a point $p_i$ to its neighbour $p_j$, and once more as the distance from $p_j$ to $p_i$, though these equal entries are not symmetric.
We will convert $D(S;m-1)$ into the distance matrix $D(S)$.
Let $d_1<d_2<\cdots<d_{m-1}$ be all distances from the first point $p_1\in S$ to all $m-1$ others.
\medskip

Each distance $d_i$ from the first row of $D(S;m-1)$ appears exactly once more in another (say, $i'$-th) row of $D(S;m-1)$.
Then $d_i$ is the distance between the points $p_1$ and $p_{i'}$ numbered as the $i'$-th row.
The map of indices $i\mapsto i'$ is a permutation of $\{2,\dots,m\}$.
We set $D_{11}=0$ and $D_{1,i'}=d_i$ for each $i=2,\dots,m$.  
Then we similarly permute indices in the 2nd row of $D(S;m-1)$, starting from the 3rd index due to the symmetry of $D(S)$, and so on.
The full distance matrix $D(S)$ uniquely determines a set with ordered points $S\subset\R^n$ modulo isometries by the classical multi-dimensional scaling in \cite[Section~8.5.1]{liberti2017euclidean}.
\end{proof}

%B=============
\section{Isometries, isometry invariants and homometric sets}
\label{sec:background}

This section covers key facts about isometries in $\R^n$ and their invariants.
We also rigorously define homometric sets, which had different interpretations in past papers \cite{patterson1944ambiguities, franklin1974ambiguities}. 

\begin{dfn}[isometry]
\label{dfn:isometry}
An {\em isometry} of $\R^n$ is a map $f:\R^n\to\R^n$ that preserves the Euclidean distance, so $|p-q|=|f(p)-f(q)|$ for any points $p,q\in\R^n$.
The map $f$ also preserves the {\em orientation} if the matrix whose columns are images under $f$ of the standard basis vectors $\vec e_1,\dots,\vec e_n$ has a positive determinant.
In this case $f$ can be called a {\em rigid motion}, because $f$ is included into a continuous family of isometries $f_{t}:\R^n\to\R^n$, $t\in[0,1]$, where $f_1=f$ and $f_0$ is the identity map $f_0(p)=p$ for any $p\in\R^n$.
\bs
\end{dfn}

Any isometry of $\R^n$ can be decomposed into at most $n+1$ reflections over hyperspaces, for example over planes in $\R^3$, hence is bijective and can be inverted.
A composition of isometries is also an isometry and defines the operation in the group $\iso(\R^n)$ of all isometries in $\R^n$.
Rigid motions are orientation-preserving isometries and form the smaller subgroup $\iso^+(\R^n)\subset\iso(\R^n)$.  
Example rigid motions in $\R^3$ are translations by vectors and rotations around straight lines.
It suffices to classify sets up to general isometries, because if we know that two point sets $S,Q$ are isometric, one can easily check if a possible isometry $S\to Q$ preserves an orientation defined as follows.
\medskip

Let a set $S\subset\R^n$ have $n+1$ points $p_0,\dots,p_n$ that are not in any $(n-1)$-dimensional subspace.
Then the determinant of the $n\times n$ matrix with columns $p_i-p_0$, $i=1,\dots, n$, is the signed volume of the parallelepiped spanned by these $n$ vectors.
The sign of this determinant can considered as an \emph{orientation} of $S$.
An isometry $f$ preserves an orientation if the determinant obtained from the points $f(p_0),\dots,f(p_n)\in Q$ has the same sign as $S$. 
\medskip

For any $n\times n$ matrix $A$, recall that $A^T$ denotes the \emph{transpose} matrix with elements $A^T_{ij}=A_{ji}$, $i,j=1,\dots,n$.
A matrix $A$ is \emph{orthogonal} if the inverse matrix $A^{-1}$ equals the transpose $A^T$. 
Orthogonality of a matrix $A$ means that $\vec v\mapsto A\vec v$ maps any orthonormal basis to another orthonormal basis.
All orthogonal matrices $A$ have the determinant $\det A=\pm 1$.
If $\det A=1$, then the map $\vec v\mapsto A\vec v$ preserves an orientation of $\R^n$.
All orthogonal matrices $A$ with $\det A=1$ form the special orthogonal group $\SO(\R^n)$, where the operation is the matrix multiplication.
The group $\SO(\R^2)$ consists of rotations about the origin in the plane.
The group $\SO(\R^3)$ consists of rotations about axes passing through the origin in $\R^3$.
All orientation-preserving isometries in $\R^n$ can be decomposed into translations and high-dimensional rotations $R\in\SO(\R^n)$ around the origin in $\R^n$.
%In general, $\SO(\R^n)$ consists of all isometries from $\iso^+(\R^n)$ that preserve the origin.
\medskip

For a given equivalence relation, any objects can be distinguished (claimed to non-equivalent) only by an invariant is preserved the given equivalence hence is independent of a representation of an object.
Surprisingly many descriptors of crystals include non-invariants, for example parameters of an ambiguous unit cell or atomic coordinates in an arbitrary basis. 

\begin{dfn}[isometry invariant]
\label{dfn:invariant}
An \emph{isometry class} of sets is a collection of all sets that are isometric to each other, i.e. any sets $S,Q$ from the same class are related by an isometry $S\to Q$.
An {\em isometry invariant} is a function $I$ that maps all sets of a certain type, e.g. all periodic sets, to a simpler space (numbers, matrices) so that $I(S)=I(Q)$ for any isometric sets $S,Q$.
An invariant $I$ is called \emph{complete} if the converse is also true: if $I(S)=I(Q)$, then the sets $S,Q$ are isometric.
\bs
\end{dfn}

Now we discuss homometric periodic sets that were hard to distinguish up to isometry, because they have identical diffraction patterns depending only on the difference set below.

\begin{dfn}[homometric sets]
\label{dfn:homometry}
For any finite set $S\subset\R^n$ of $m$ points, the \emph{difference} multi-set $\dif(S)$ consists of the $m^2$ vector differences $\vec a-\vec b$ for all points $a,b\in S$, counted with multiplicities.
Periodic point sets $S,Q\subset\R^n$ with a common lattice $\La$ and a primitive unit cell $U$ are called \emph{homometric} if $\dif(S\cap U)\equiv\dif(Q\cap U)\pmod{\La}$ with multiplicities respected, so all pairs of vectors $u\in\dif(S\cap U)$ and $v\in\dif(Q\cap U)$ that are equal up to lattice translations have the same multiplicity in both difference sets.
The \emph{homometry} is the equivalence relation of periodic point sets being homometric in the sense above.
\bs
\end{dfn}

If a set $S$ consists of $m$ points, $\dif(S)$ includes the zero vector with multiplicity $m$.
The above definition clarifies the following past attempts to define homometric sets.
\medskip

Patterson \cite[p.197]{patterson1944ambiguities} called periodic point sets $S,Q\subset\R^n$ \emph{homometric} if $\dif(S\cap U)\equiv\dif(Q\cap U)\pmod{\La}$ without mentioning weights or multiplicities.
Franklin \cite[equations (17)-(18)]{franklin1974ambiguities} renamed them as \emph{homometric modulo a lattice} $\La$ and called $S,Q$ \emph{homometric} if $\dif(S\cap U)=\dif(Q\cap U)$, not modulo $\La$.
Additionally both definitions required that $S,Q$ are not isometric.
However, after removing this restriction, we expect to get an equivalence relation so that any periodic point set should be homometric to itself even if another unit cell of a lattice $\La$ is chosen.
The following example shows that the equation $\dif(S\cap U)=\dif(Q\cap U)$ without translations by the lattice $\La$ fails this reflexivity condition.
\medskip

Franklin \cite[p.~699]{franklin1974ambiguities} considered the sets $S_3=\{0,1\}+3\Z$ and $Q_3=\{0,2\}+3\Z$, which are isometric by $x\mapsto x+1\pmod{3}$.
However, $\dif(\{0,1\})=\{0,0,-1,+1\}\neq\dif(\{0,2\})=\{0,0,-2,+2\}$.
These sets are equal modulo 3, hence lattice translations are needed.
If we consider the set $S_3$ with a twice larger unit cell  and period 6 as $\{0,1,3,4\}+6\Z$, then 
$$\dif(\{0,1,3,4\})=\{0,0,0,0,\pm 1,\pm 1,\pm 2,\pm 3,\pm 3,\pm 4\}.$$

This difference set can be considered equal to $\dif(\{0,1\})$ modulo 3 only if the multiplicities in both sets are normalized so that their sums are equal.
Hence Definition~\ref{dfn:homometry} requires a primitive unit cell $U$.
Most importantly, Proposition~\ref{prop:homometry} below justifies that homometry in Definition~\ref{dfn:homometry} is independent of a primitive unit cell and is an \emph{equivalence} relation satisfying %all the axioms: 
\smallskip

\noindent
(1) \emph{reflexivity} : a periodic set $S$ is homometric to $S$, i.e. $S\sim S$;
\smallskip

\noindent
(2) \emph{symmetry} : if $S$ is homometric to $Q$, i.e. $S\sim Q$, then $Q\sim S$;
\smallskip

\noindent
(3) \emph{transitivity} : if $S\sim Q$ and $Q\sim T$, then $S\sim T$.
\medskip

Proposition~\ref{prop:homometry} makes the experimental concept of homometric crystals verifiable in an algorithmic way.
It might be a folklore result, but we couldn't find a proof in the literature.

\begin{prop}[algorithm for homometric sets]
\label{prop:homometry}
(a) For any periodic point set $S\subset\R^n$ with a lattice $\La$, the difference set $\dif(S\cap U)\pmod{\La}$ does not depend on a primitive unit cell $U$ of $S$.
So the homometry in Definition~\ref{dfn:homometry} is an equivalence relation.
\medskip

\noindent
(b) Given a common primitive unit cell $U$ containing $m$ points of periodic point sets $S,Q\subset\R^n$, there is an algorithm of complexity $O(m^2 \log m)$ to determine if $S,Q$ are homometric.
\bs
\end{prop}
\begin{proof}
Let $U,U'$ be primitive cells of the periodic set $S\subset\R^n$.
Any point $q\in S\cap U'$ can be translated along a vector $\vec v\in\La$ to a point $p\in S\cap U$ and vice versa.
These translations establish a bijection $S\cap U\lra S\cap U'$, which can change any point only by a vector of $\La$.
So $\dif(S\cap U)\equiv\dif(S\cap U')\pmod{\La}$ with multiplicities respected by the bijection above.
\medskip

To determine if periodic sets $S,Q\subset\R^n$ are homometric by Definition~\ref{dfn:homometry}, first we compute all $O(m^2)$ pairwise vector differences between points from the motifs $S\cap U$ and $Q\cap U$.
\medskip

To check if these vector sets coincide, we could lexicographically order them in time $O(m^2\log m)$, e.g. by using coordinates in the basis of the cell $U$.
Then a single pass over $O(m^2)$ vector differences is enough to decide if $\dif(S)\equiv\dif(Q)\pmod{\La}$.
\end{proof}

We illustrate Proposition~\ref{prop:homometry} for the 1-dimensional periodic sets $S(1)=\{0, 1, 3, 4\}+8\Z$ and $Q(1)=\{0, 3, 4, 5\}+8\Z$ in Fig.~\ref{fig:non-isometric_pairs}.
Their 4-point motifs have distinct difference sets: \\
$\begin{array}{l|llll}
S_8 &  0 & 1 & 3 & 4 \\
\hline
0 & 0 & -1 & -3 & -4 \\
1 & 1 & 0 & -2 & -3 \\
3 & 3 & 2 & 0 & -1 \\
4 & 4 & 1 & 3 & 0
\end{array}$ and 
$\begin{array}{l|llll}
Q_8 &  0 & 3 & 4 & 5 \\
\hline
0 & 0 & -3 & -4 & -5 \\
3 & 3 & 0 & -1 & -2 \\
4 & 4 & 1 & 0 & -1 \\
5 & 5 & 2 & 1 & 0
\end{array}$.
These sets coincide modulo 8 with multiplicities shown as subscripts:
$\dif(S(1))\equiv \{0_4,1_2,2_1,3_2,4_2,5_2,6_1,7_2 \}\equiv \dif(Q(1))$.
\medskip

Then $S(1),Q(1)$ are homometric by Definition~\ref{dfn:homometry}.
The equivalence of differences modulo 8 gives rise to a bijection between all 16 elements of the matrices above, hence to a bijection between the sets of vector differences $\dif(S(1))\to\dif(Q(1))$.
For example, the difference $(8i+1)-(8j+4)=8(i-j)-3\equiv 5\pmod{8}$ in $S(1)$ can be bijectively mapped to the vector difference $(8i+5)-8j=8(i-j)+5$ in $Q(1)$.

%C=============
\section{Comparisons with crystallographic tools for experts in applications}
\label{sec:comparisons}

This paper makes a step towards the computer-aided design of solid crystalline materials, which are also referred to as \emph{crystals}. 
Most solid minerals in nature and a number of important synthetic materials are periodic crystals at the atomic scale.
\medskip

A periodic point set even without extra links (bonds) is the fundamental model for any solid crystalline material.
Indeed, atomic centers (nuclei) have a less ambiguous physical meaning than inter-atomic bonds, which can be challenging to classify and that require bespoke and empirical definitions.
For example, at what distance does a hydrogen bond become a bond?
However, for any zero-size point representing an atom, one can add a  label such as a chemical element or a radius or another physical property.
\medskip

Definition~\ref{dfn:AMD} is given for unlabeled points only for simplicity.
If points have labels such as chemical elements, one can adapt AMD for labels as follows.
For any two labels (say, O for oxygen and N for nitrogen) and $i$-th point $p_i$ of type O, the matrix $D[\text{O}-\text{N}](S;k)$ has the $i$-th row of ordered distances from $p_i$ to its $k$-th nearest neighbours of type N within $S$.
Define the label-specific $\AMD[\text{O}-\text{N}](S;k)$ as the average of the $k$-th column of the matrix $D[\text{O}-\text{N}](S;k)$ over all points of type O in the motif of $S$. 
The right hand side picture of Fig.~\ref{fig:AMD} shows AMD values without taking labels into account.
We can include labels as follows. 
The green unit cell has two points: one red and one blue (traditional colors for oxygen and nitrogen, respectively).
Then the matrix $D[\text{O}-\text{N}]$ is the single row-vector $\AMD[\text{O}-\text{N}]$ of ordered distances from the red point to all its blue neighbours.
\medskip

The idea of the Average Minimum Distances is close to the Radial Distribution Function (RDF), which measures the probability to find a number of atoms at a distance $r$ from a reference atom \cite[vol.~B, sec.~4.6]{hahn1983international}.
The RDF is defined via a multidimensional integral, which is approximated and visualized as a smooth function. 
Homometric crystals have identical RDFs but can be distinguished by AMD, see Example~\ref{exa:SQr+SQ15}.
%Due to Theorem~\ref{thm:asymptotic}, if periodic sets have the same point packing coefficient, the beginning of an AMD vector can be enough to distinguish them up to isometry as shown in Fig.~\ref{fig:lattices2D_AMD30+300}.
\medskip

Diamond and graphite are well-known crystals in Fig.~\ref{fig:graphene+diamond+T2delta} with identical compositions (pure carbon) but very different physical properties, so geometric structures are important.
\medskip

\begin{figure}[h!]
\includegraphics[height=40mm]{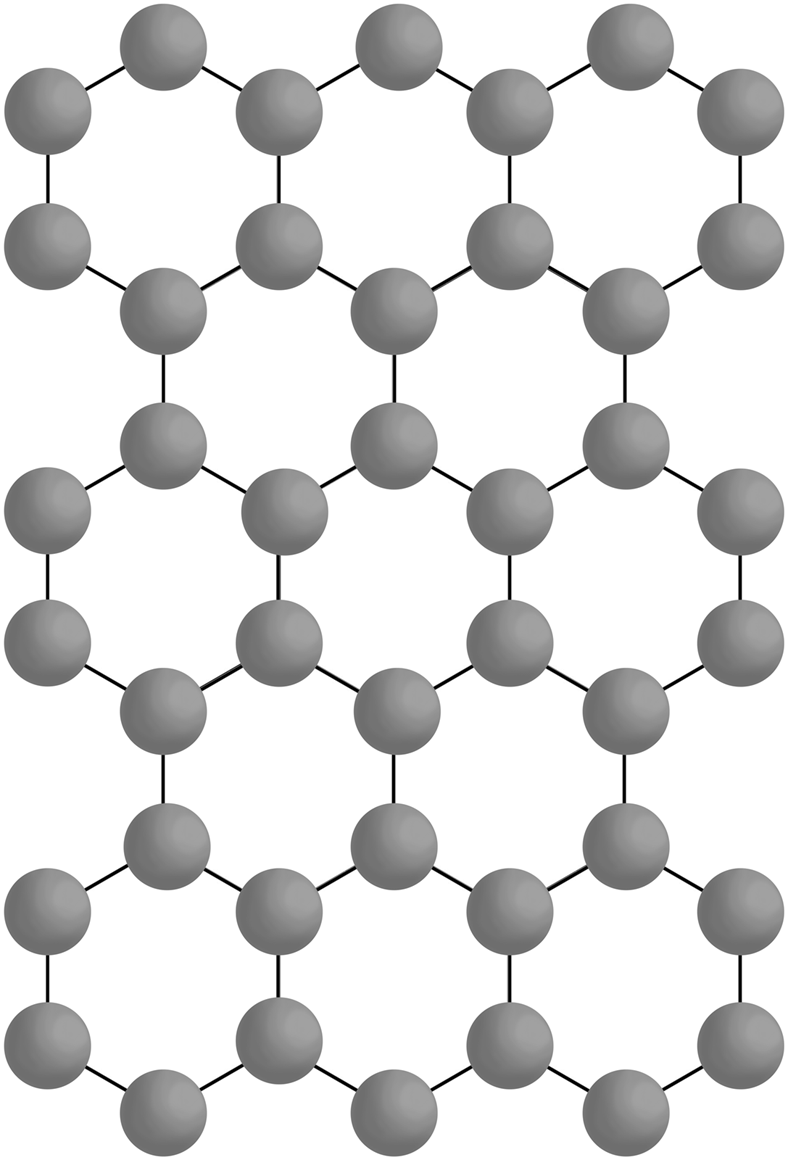}
\hspace*{8mm}
\includegraphics[height=40mm]{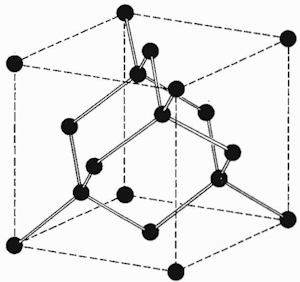}
\hspace*{8mm}
\includegraphics[height=40mm]{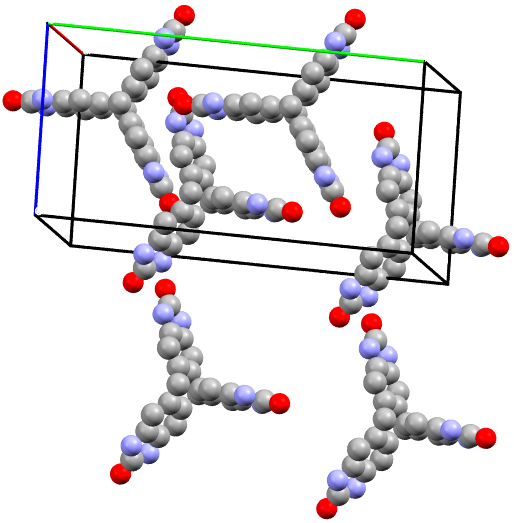}
\caption{Graphite (left) and diamond (middle) consist of only carbon atoms, hence their physical properties differ due to crystal structures.
The crystal on the right consists of T2 molecules.}
\label{fig:graphene+diamond+T2delta}
\end{figure}

Typically, periodic crystal is stored as a Crystallographic Information File (CIF) including parameters of a unit cell (three edge-lengths and three angles in $\R^3$) and all atoms with chemical labels and coordinates of atomic centers in the basis of the unit cell.
The input size of a crystal is best measured as the number of atoms in its unit cells or the number of points in a motif of a periodic set, which motivates a near linear time condition in Problem~\ref{pro:invariants}.
\medskip

Crystallographers used Niggli's reduced cell  \cite[section~9.3]{hahn1983international} as a canonical cell for a lattice.
This reduced cell is unique, but at the price of continuity.
Examples of discontinuous Niggli's cells \cite{andrews1980perturbation} were known since 1980.
%Similar examples are known for other types of reduction by Voronoi or Selling \cite{conway1992low}.
Theorem~\ref{thm:discontinuity} disproves any continuous reduction.
\medskip

A traditional method to check if two crystals are identical compares their finite motifs of atoms in reduced cells.
Niggli's cell \cite[section~9.3]{hahn1983international} is unique, but is known to be discontinuous \cite{andrews1980perturbation} in the sense that a reduced cell of a slightly perturbed lattice is not close to a reduced cell of the non-perturbed lattice.
Discontinuity under small atomic perturbations is the major weakness of all discrete invariants including symmetry groups.
The nearly identical periodic sets in the last two pictures of Fig.~\ref{fig:deformations} should be recognisable as very similar.
\medskip

Though there was no justified distance that satisfies all metric axioms for any periodic crystals, the COMPACK algorithm \cite{chisholm2005compack} is widely used for a pairwise comparison of crystals.
Within given tolerances ($20^{\circ}$ for angles and 20\% for distances), up to a given number (15 by default) of molecules from two crystals are matched by a rigid motion that minimizes the Root Mean Square
 deviation of $N$ matched atoms $\RMS=\sqrt{\dfrac{1}{N}\sum\limits_{i=1}^N|p_i-q_i|^2}$. 

% horizontal table 
% \medskip
\begin{table*}[h!]
\begin{tabular}{c|cccccccccc}
$m$ matched molecules & 5 of 5 & 9 of 10 & 12 of 15 & 16 of 20 & 21 of 25 & 26 of 30 & 28 of 35 %& 31 of 40 & 35 of 45 & 40 of 50 
\\
\hline
RMS, $1\AA=10^{-10}m$ & 0.603 & 0.708 & 0.874 &  0.969 & 1.080 & 1.040 & 1.044 %& 1.178 & 1.290 & 1.301 
\\
\hline
running time, seconds & 0.168 & 0.422 & 2.026 & 14.61 & 63.51 & 151.4 & 759.3 %& 1260 & 1298 & 1466
\end{tabular}
\caption{The Root Mean Square (RMS) deviation between the experimental T2-$\delta$ crystal \cite{pulido2017functional} and its closest simulated version with ID 14.
The irregular dependence of RMS on $m$ makes this comparison unreliable.
%Experiments on the same machine show 
The running time substantially grows in the number of molecules.}
\label{tab:T2rms}
\end{table*}

Informally, the RMS is the average version of the bottleneck distance $d_B$ restricted to a finite subset of atoms or molecules whose choice may depend on extra parameters.
If a match between these finite subsets is extended, $\RMS$ can also increase, potentially to infinity, for example between cubic lattices of sizes 1 and 1.1. 
Table~\ref{tab:T2rms} shows how RMS depends on the maximum number $m$ of attempted molecules to match by rigid motion.
\medskip

%D=========================
\section{Breakthrough invariant-based visualisations of huge datasets}
\label{sec:visualisations}

The past visualisation was possible only for datasets of simulated crystals (based on the same chemical compositions) and only by a few explainable coordinates, see a noisy plot without any internal similarities in Fig.~\ref{fig:CSP}.
The AMD-based visualisation shows the same dataset as a Minimum Spanning Tree in Fig.~\ref{fig:T2AMD200TMap} with all  crystals clearly distinguishable.
\medskip

The AMD invariants have captured the structure-property relationships
in the map of all crystalline drugs in Fig.~\ref{fig:CSD_drugs_TMap_AMD200} showing aspirin and paracetamol close to each other as in a pharmacy.
Finally, the ultra-fast AMD allowed us to produce the largest ever map of all 229K organic crystals from the Cambridge Structural Database in Fig.~\ref{fig:CSDorganic_AMD100TMap}.  
All these maps can be interactively zoomed and explored in any Javascript-enabled browser.
\medskip

The practical motivation for continuous invariants of periodic point sets is the `embarrassment of over-prediction' %\cite{price2018zeroth} 
coined by Prof Sally Price FRS.
Modern tools of Crystal Structure Prediction (CSP) predict many thousands of simulated crystals, though very few of them can be experimentally realized.
CSP is often run on supercomputers and starts from millions of random arrangements of atoms or molecules that are then minimized according to a complicated lattice energy function.
Lower lattice energy values indicate that crystals are more likely to remain stable as solid materials. 
This energy has no closed expression due to the dependence on infinitely many interactions between all atoms within a periodic crystal.
\medskip
 
A typical CSP software outputs thousands of approximate local minima of this energy, that is simulated crystals whose local perturbations are unlikely to produce more stable arrangements.
Simulated crystals are often visualized by the energy-vs-density \emph{landscape} representing each crystal as a point with two coordinates (density, energy).
The density is the molecular weight of atoms within a cell divided by the cell volume.
This single-value density is insufficient to differentiate crystals, because many non-isometric sets can have the same density.
The landscape in Fig.~\ref{fig:CSP} required 12 weeks of supercomputer time \cite{pulido2017functional}.
%One potential application of isometry invariants is to accelerate materials discovery by filtering near duplicates in CSP datasets.
%Then much less time will be spent on further simulations of target properties. % such as gas adsorption or solubility. 
\medskip

\begin{figure}[h]
\centering
\includegraphics[width=\linewidth]{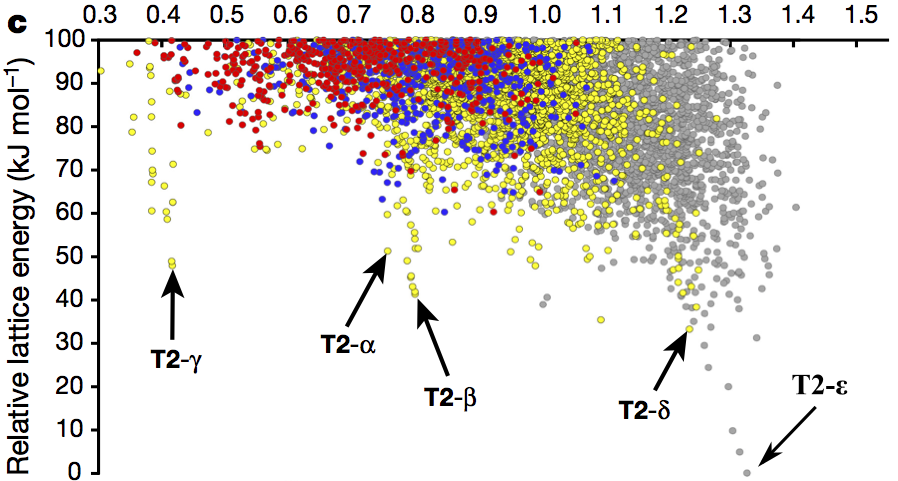}
\caption{\cite[Fig.~2d]{pulido2017functional}:
energy-vs-density plot shows many of 5679 crystals as nearly identical. 
}
\label{fig:CSP}
\end{figure}
% \vspace*{-4mm}

Laboratory syntheses were guided by the energy landscape in Fig.~\ref{fig:CSP} of 5679 simulated crystals.
To validate this approach, each experimental crystals is matched with the 5679 simulated structures.
If there was no close match, the simulations have missed a real crystal, which is possible, because the continuous space of all potential crystals in Fig.~\ref{fig:CSP} is blindly sampled.
Until now, only the single-value density $\rho$ was 
most commonly used as the continuous isometry invariant to match similar crystals.
The density $\rho$ can separate well only nano-porous organic crystals, while inorganic crystals are much denser and can not be well-separated.
Using the density $\rho$ of an experimental crystal, one can search through multiple simulated crystals within a vertical `stripe' of the energy-vs-density landscape in Fig.~\ref{fig:CSP} over a small density interval to allow for errors.
From this stripe one might take the simulated crystal with the lowest energy as the most likely structure.
As such, a result depends on the tolerance error for the density among other factors. 
A final match is confirmed by the non-invariant RMS deviation between finite portions of crystals, see Table~\ref{tab:T2rms}.
\medskip

\begin{figure}[h]
\centering
\includegraphics[height=22mm]{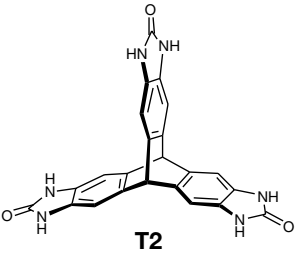}
\hspace*{0mm}
\includegraphics[height=22mm]{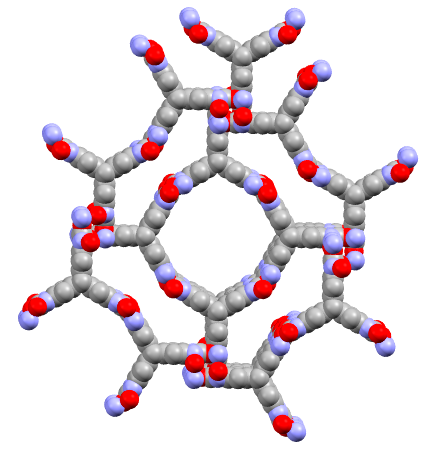}
\hspace*{0mm}
\includegraphics[height=22mm]{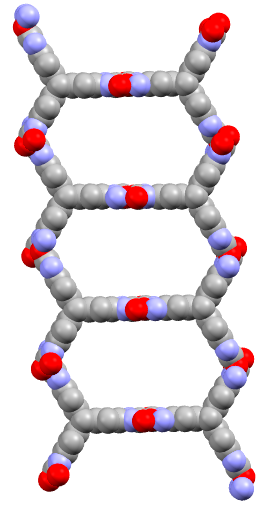}
\hspace*{0mm}
\includegraphics[height=22mm]{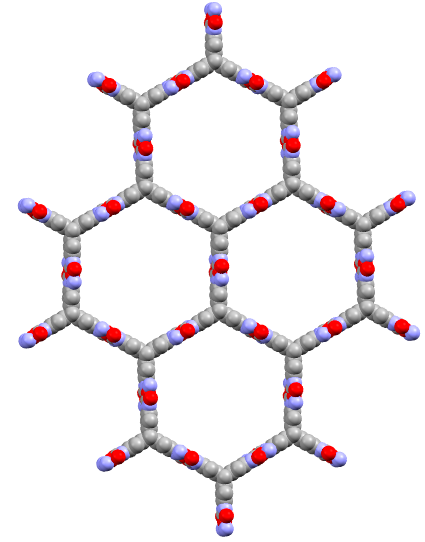}
\hspace*{0mm}
\includegraphics[height=22mm]{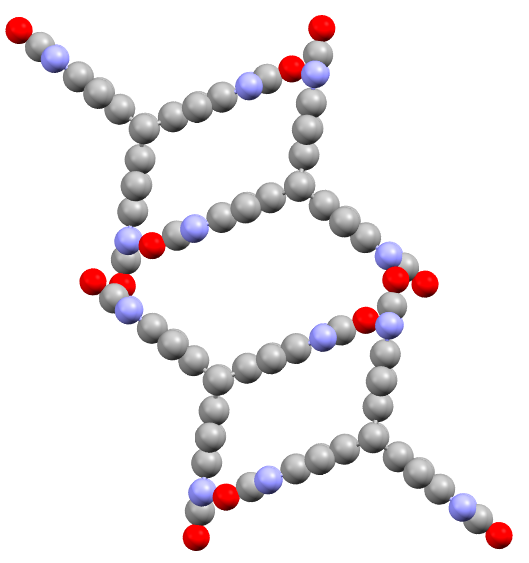}
\hspace*{0mm}
\includegraphics[height=22mm]{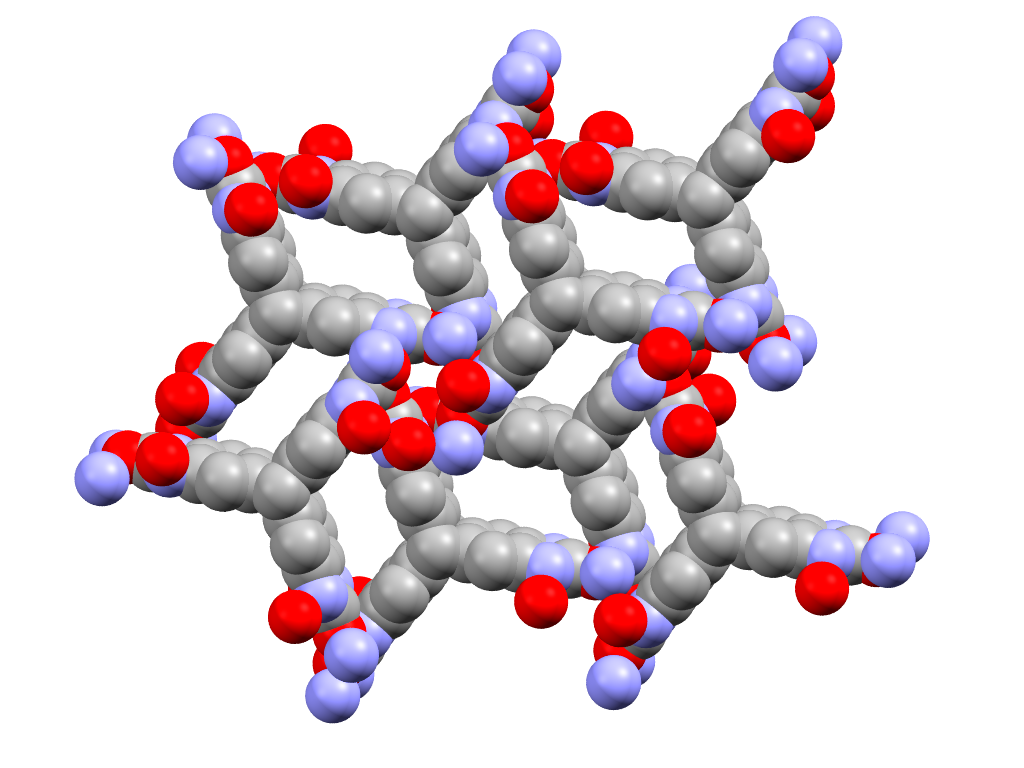}
\caption{
T2 molecule and the crystals T2-$\al$, T2-$\be$, T2-$\ga$, T2-$\de$, T2-$\ep$ based on the T2 molecule were synthesized in the laboratory after the Crystal Structure Prediction in Fig.~\ref{fig:CSP} reported in \cite{pulido2017functional}.}
\label{fig:T2crystals}
\end{figure}
% \vspace*{-6mm}

The following closest simulated versions of experimental crystals were 
found previously by manual (and time-consuming) searches, but can now be automatically identified as close neighbours in the clustering dendrogram in Fig.~\ref{fig:AMD1000dendrogram_Linf}.
Here are the numerical IDs in the T2 dataset: 0186 for T2-$\al$, 0054 for T2-$\be$, 0120 for T2-$\ga$, 0014 for T2-$\de$, 0001 for T2-$\ep$.
Fig.~\ref{fig:AMD1000dendrogram_Linf} shows the dendrogram obtained by the $L_{\infty}$-distance, which is more stable with respect to the length $k$ of AMD vectors than Euclidean $L_2$ due to Theorem~\ref{thm:asymptotic}.
This dendrogram automatically identifies closest simulated crystals that in the past were manually matched to experimental ones 
within a prescribed density range.
%\medskip
All expected similar crystals, such as four versions of T2-$\ga$ synthesized under different conditions, belong to the same small cluster highlighted by the orange color in Fig.~\ref{fig:AMD1000dendrogram_Linf}.
Fig.~\ref{fig:T2AMD200curves} clearly shows five different patterns for the various experimental structures of T2, thus showing that AMD curves can distinguish between polymorphs (that is, different crystal packings of the same molecule).  

\begin{figure*}
\centering
\includegraphics[width=0.85\textwidth]{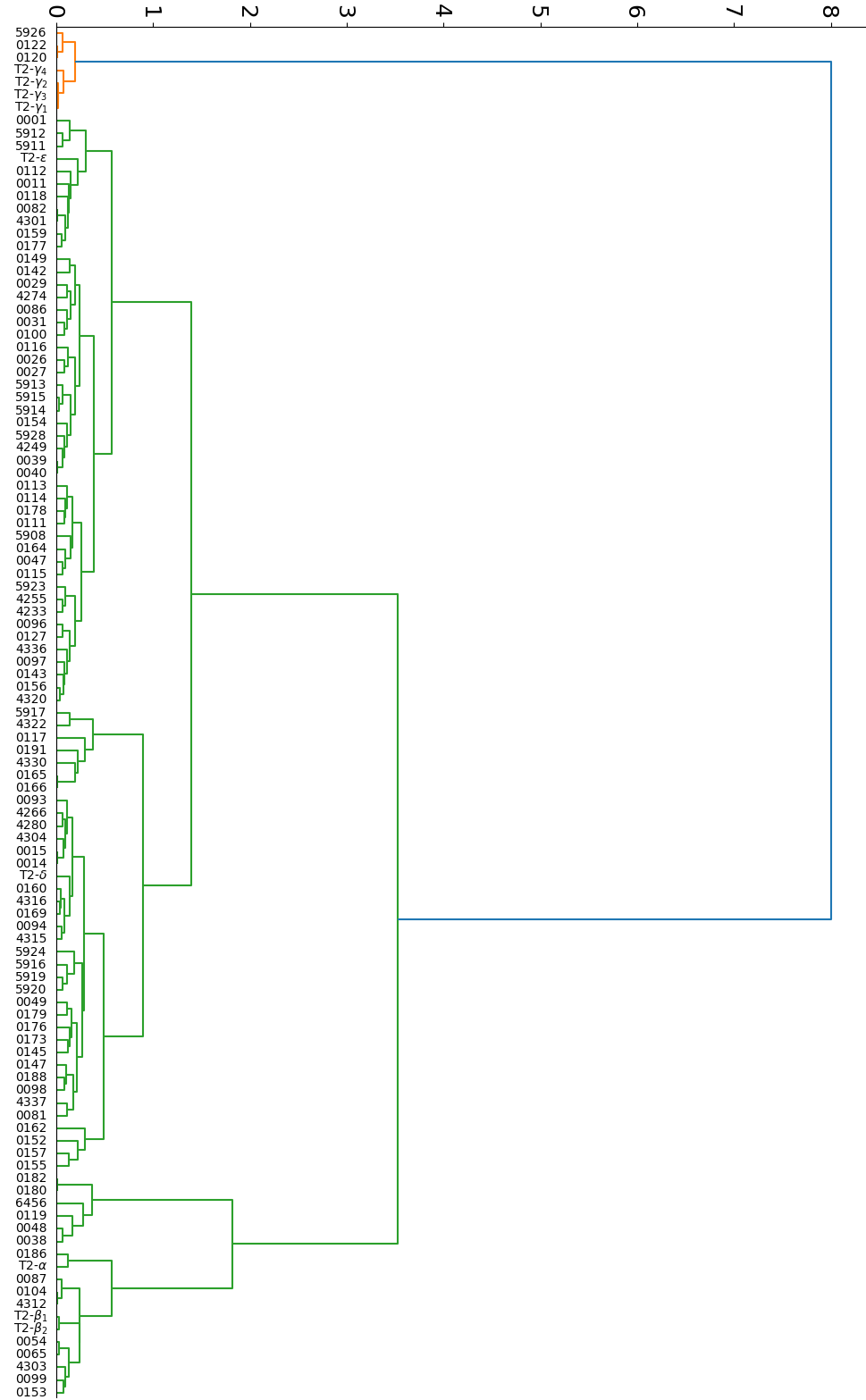}
\caption{Complete-linkage clustering by the $L_{\infty}$-distance on $\AMD^{(1000)}$ of T2 crystals: nine experimental and 100 simulated crystals with lowest energies, see \cite{pulido2017functional}.}
\label{fig:AMD1000dendrogram_Linf}
\end{figure*}

We note that there are four structures reported in the Cambridge Structural Database for polymorph T2-$\ga$, and that these all cluster onto the same AMD curve in Fig.~\ref{fig:T2AMD200curves}.
The linear behavior of $(\AMD_k(S))^3$ in the top diagram is explained by Theorem~\ref{thm:asymptotic}.
% and the point packing coefficient $c(S)$ in Table~\ref{tab:densities}. 
%Even after the normalization by $c(S)$, the graphs $\left(\dfrac{\AMD_k(S)}{c(S)}\right)^3$ in the lower diagram still show the differences between five families.
\medskip

\begin{figure*}[h!]
\includegraphics[width=\linewidth]{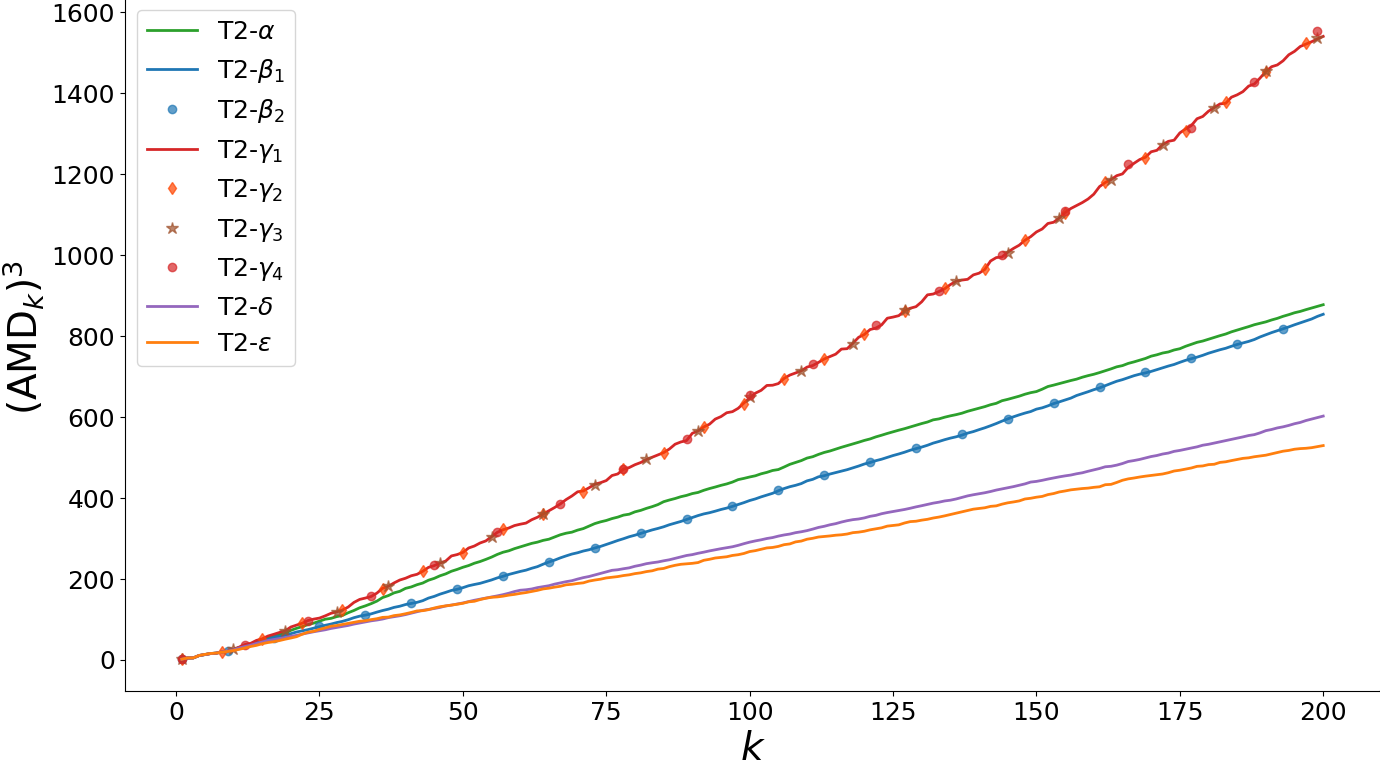}
\caption{
Cubed AMD curves clearly distinguish five phases of nine experimental crystals.}
\label{fig:T2AMD200curves}
\end{figure*}

The AMD curves of the four T2-$\ga$ crystals seem identical in Fig.~\ref{fig:T2AMD200curves}.
The AMD invariants pick up their subtle differences in the \emph{AMDigrams} similar to barcodes as shown in Fig.~\ref{fig:T2AMDigrams}.
The values of $\AMD_k$ for $k\leq 10$ are similar because all T2 crystals have the same intramolecular bond distances; small values of $k$ correspond to the atomic distances within the T2 molecule.
However, larger values distinguish geometric structures of different crystals. 
Theorem~\ref{thm:asymptotic} implies that pushing $k$ to infinity will add little extra information.
For reference, our $k=200$ is about 5 times more than the number 46 of atoms in a single T2 molecule in Fig.~\ref{fig:T2crystals}.
\medskip

\begin{figure*}[h!]
\includegraphics[width=\linewidth]{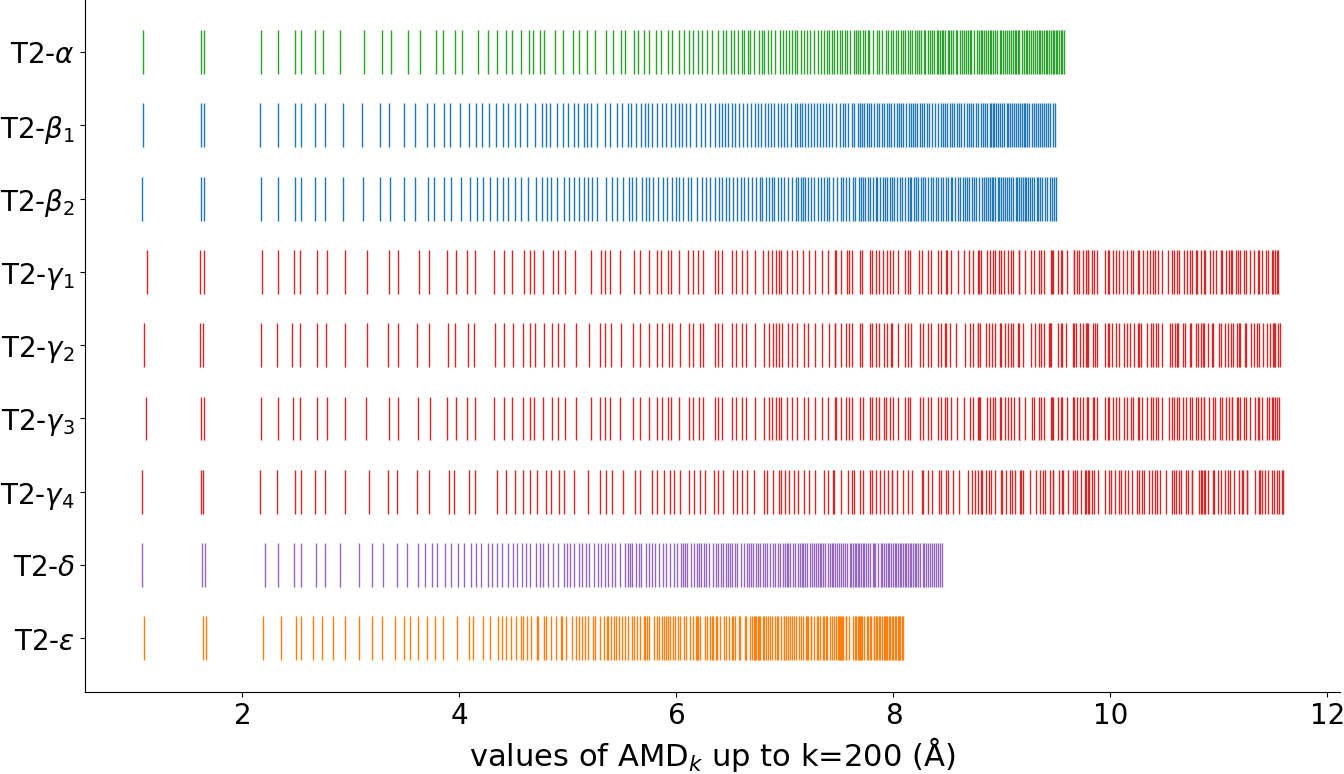}
\caption{
The \emph{AMDigrams} $\AMD^{(k)}(S)$ of the nine experimental crystals 
%from Table~\ref{tab:densities} 
consist of vertical bars drawn at the values $\AMD_1,\dots,\AMD_k$ in the horizontal axis.
Even subtle differences are now visible.}
\label{fig:T2AMDigrams}
\end{figure*}

CSP datasets are often visualized as energy-vs-density plots, because the density is one of the few quickly computable continuous isometry invariants, see Fig.~\ref{fig:CSP}. 
However, the single-value density $\rho$ is insufficient to differentiate between many similar crystals.
The energy is computationally predicted only for simulated crystals, not for experimental ones.
%crystals are not included in energy-vs-density plots.
%; whereas the energy measurement of a real crystal requires an explosion to disintegrate the material, which is unfavourable for obvious reasons. 
%Hence energy-vs-density plots cannot really show if or how experimental crystals are close to their simulated versions.
\medskip

\begin{figure*}
\includegraphics[width=\textwidth]{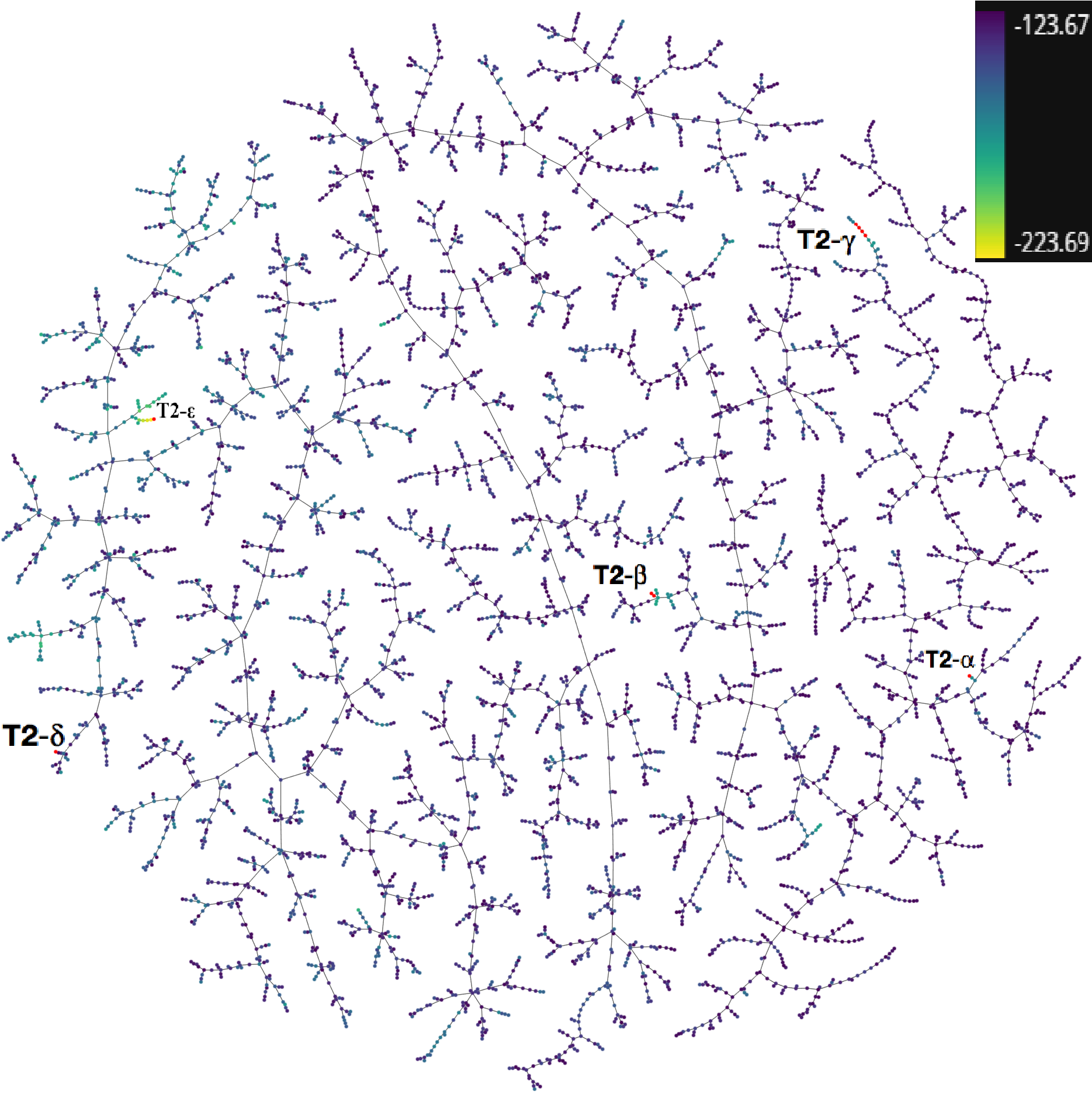}
\caption{The first joint visualization of 9 experimental and 5679 simulated crystals by TreeMap \cite{probst2020visualization} of a Minimum Spanning Tree based on $L_{\infty}$-distances between $\AMD^{(200)}$ vectors.
The colors reflect energies of simulated crystals.
Low energy values correspond to the most stable crystals in lighter colors.
The nine experimental crystals of unknown energies are shown in red for easy visibility.
We automatically 
located the closest simulated structures for all experimental crystals by using this Minimum Spanning Tree built by using computationally fast AMD invariants due to Theorem~\ref{thm:algorithm}.
%The four T2-$\ga$ crystals are close to each other.
}
\label{fig:T2AMD200TMap}
\end{figure*}

Fig.~\ref{fig:T2AMD200TMap} is a first visualization of experimental and simulated crystals together.
We computed the $L_{\infty}$-distance between the vectors $\AMD^{(200)}$ of all crystals, then built a Minimum Spanning Tree that connects all crystals as points in $\R^{200}$ and has a minimum total length of edges, drawn by TreeMap \cite{probst2020visualization}.
The points in Fig.~\ref{fig:T2AMD200TMap} are colored by predicted energy, apart from the 9 experimental crystals, where the points are coloured red.  
All experimental structures are located adjacent to closely-matched predicted structures in the T2 dataset.
\medskip

\begin{figure*}
\centering
\includegraphics[width=\textwidth]{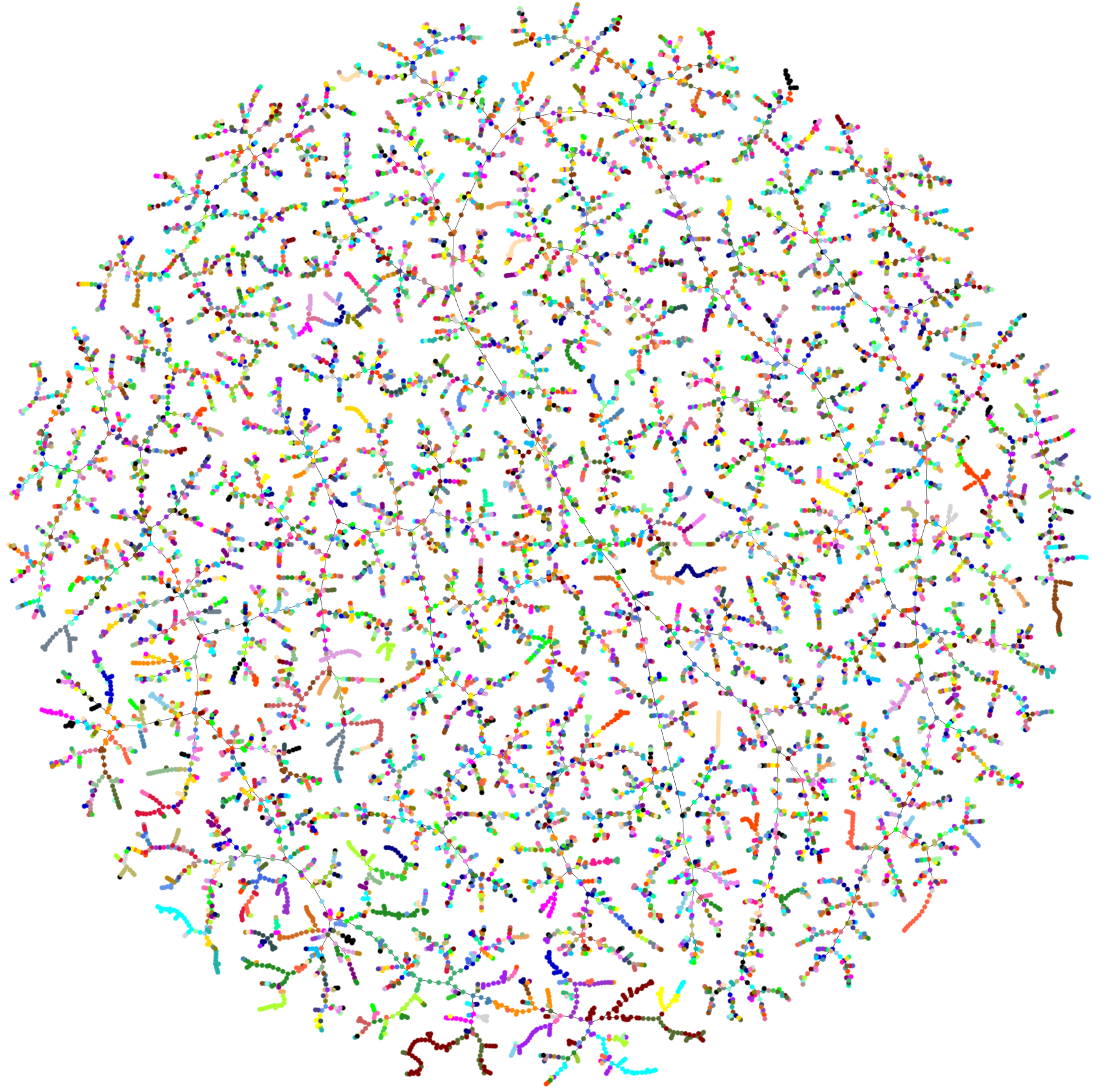} \\
\includegraphics[width=0.75\textwidth]{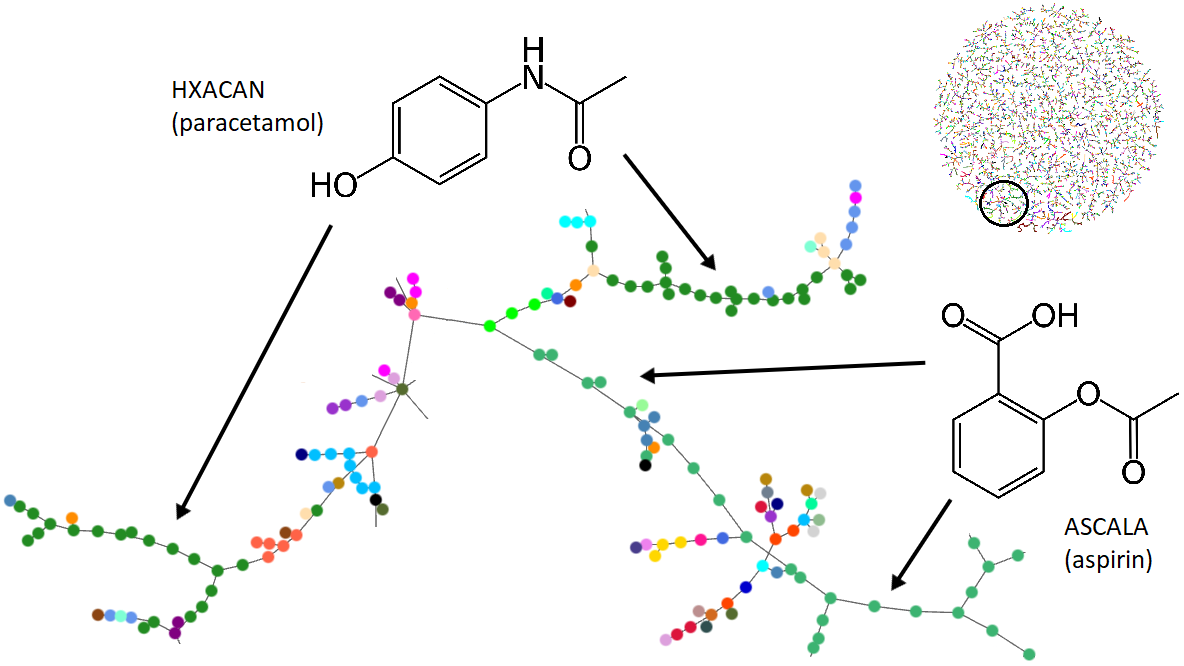}
\caption{TreeMap of all 12576 crystal structures in the CSD Drug subset, based on $L_{\infty}$-distances between $\AMD^{(200)}$ vectors.
All drugs with identical 6-letter codes have the same (random) color.
The aspirin and paracetamol are close as in a pharmacy, though the comparisons used only geometry.
}
\label{fig:CSD_drugs_TMap_AMD200}
\end{figure*}

%You can find aspirin and paracetamol on the same shelf in a pharmacy due to their common therapeutic properties.
%We have found them close to each other in the TreeMap of all 12576 crystal structures from the CSD drug subset inFig.~\ref{fig:CSD_drugs_TMap_AMD200}.

\begin{figure*}
%\centering
\includegraphics[width=1.2\textwidth]{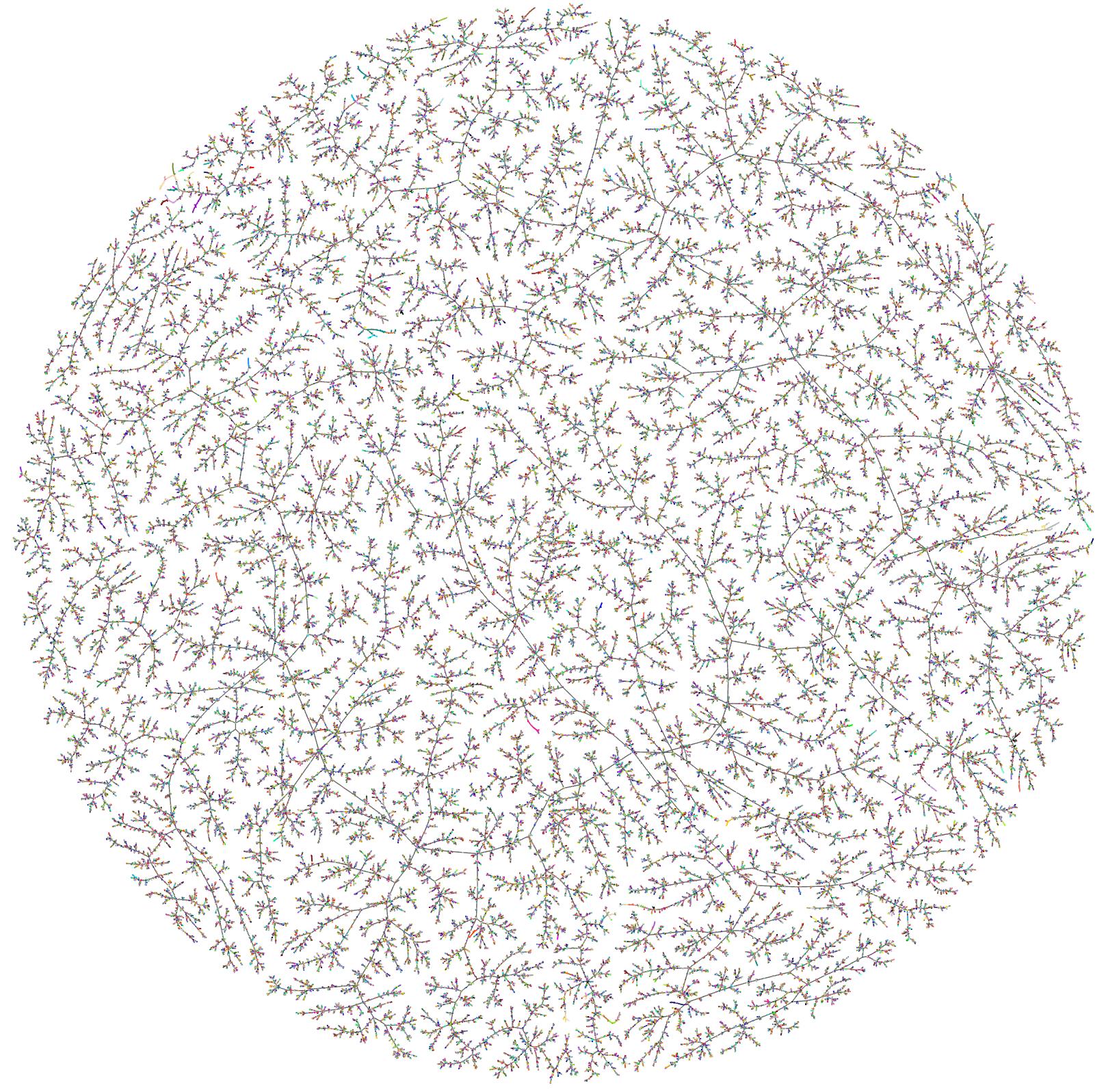}
\caption{TreeMap of all 228,994 organic crystals from the Cambridge Structural Database.}
\label{fig:CSDorganic_AMD100TMap}
\end{figure*}

\end{document}